\documentclass[journal,12pt,onecolumn,draftclsnofoot]{IEEEtran}
\usepackage{cite}  
\usepackage{color}
\usepackage{algorithm} %format of the algorithm 
\usepackage{algorithmic} %format of the algorithm 
\usepackage{multirow} %multirow for format of table 
\usepackage{amsmath} 
\usepackage{xcolor}

\usepackage{graphicx}
\usepackage{epstopdf}

\usepackage{amsfonts}
\usepackage{stfloats}
\usepackage{amsthm,amsmath,amssymb}
\usepackage{mathrsfs}
\usepackage{hyperref}
\usepackage{bm}
\usepackage{amsmath}

\newtheorem{lemma}{Lemma}
\newtheorem{proposition}{Proposition}
\ifCLASSINFOpdf

\else

\fi

\begin{document}
\title{Multiple Intelligent Reflecting Surfaces Assisted Cell-Free MIMO Communications}

\author{Xie Xie,~\IEEEmembership{Student Member,~IEEE,}
        Chen He,~\IEEEmembership{Member,~IEEE,}
        Xiaoya Li,
       % Xiaojiang Chen,~\IEEEmembership{Member,~IEEE,}
        %~\IEEEmembership{Life~Fellow,~IEEE}%
        Kun Yang,~\IEEEmembership{Senior Member,~IEEE,}
        and Z. Jane Wang,~\IEEEmembership{Fellow,~IEEE}
        %<-this % stops a space
\thanks{Xie Xie, Chen He, and Xiaoya Li are with the School of Information Science and Technology, Northwest University, Xi'an, 710069, China. Kun Yang is with the School of Information Science and Technology, Northwest University, Xi'an, 710069, China, and he is also with School of Computer Science \& Electronic Engineering, University of Essex, Wivenhoe Park, Colchester, Essex CO4 3SQ, United Kingdom. Z. Jane Wang is with Department of Electrical and Computer Engineering, The University of British Columbia, Vancouver, BC V6T1Z4, Canada. Corresponding author: Chen He (email: chenhe@nwu.edu.cn).}
%\thanks{Kun Yang is with the School of Information Science and Technology, Northwest University, Xi'an, 710069, China, and he is also with School of Computer Science \& Electronic Engineering, University of Essex, Wivenhoe Park, Colchester, Essex CO4 3SQ, United Kingdom.}
%\thanks{Z. Jane Wang is with Department of Electrical and Computer Engineering, The University of British Columbia, Vancouver, BC V6T1Z4, Canada.}
%\thanks{This work was supported  in part by the  under Grant 61701401, the New Star Program of Science and Technology in Shaanxi Province under Grant 2019KJXX-061.}
}

\markboth{Submitted to IEEE Trans on Wireless Communications}
{Shell \MakeLowercase{\textit{et al.}}: Bare Demo of IEEEtran.cls for IEEE Journals}

\maketitle

\begin{abstract}
\textcolor{blue}{
In this paper, we investigate an intelligent reflecting surface (IRS) assisted cell-free multiple input multiple output (MIMO) communication system, where distributed multiple IRSs are dedicated deployed to assist distributed multiple base stations (BSs) for cooperative transmission. Our objective is to maximize the achievable sum-rate of the cell-free system by jointly optimizing the active transmit beamforming matrices at BSs and the passive reflecting beamforming matrices at IRSs, while the constraints on the maximum transmit power of each BS and the phase shift of each IRS element are satisfied. We propose an efficient framework to jointly design the BSs, the IRSs, and the user equipment (UEs). 
As a compromise approach, we first transform the non-convex problem into an equivalent form based on the fractional programming methods and then decompose the reformulated problem into two subproblems and solve them alternately. Particularly, we propose a Lagrangian dual sub-gradient based algorithm to solve the subproblem of optimizing the active transmit beamforming with nearly closed-form solutions. We reformulate the subproblem of optimizing the passive reflecting beamforming as a constant modulus constrained quadratic programming (CMC-QP) problem. We first provide two feasible solutions by proposing a pair of relaxation-based algorithms. We also develop a low-complexity alternating sequential optimization (ASO) algorithm to obtain closed-form solutions. All three algorithms are guaranteed to converge to locally optimal solutions. 
Simulation results demonstrate that the proposed algorithms achieve considerable performance improvements compared with the benchmark schemes.}
\end{abstract}

\begin{IEEEkeywords}
Intelligent Reflecting Surface, Cell-Free, Constant Modulus Constraint, Quadratic Programming.
\end{IEEEkeywords}

\IEEEpeerreviewmaketitle

\section{Introduction}
Densely deployed base stations (BSs) have been proposed to support massive wireless devices and satisfy high rate expectations.
The multiple cellular (multi-cell) technique is a conventional method to cover a large area with multiple user equipment (UEs) \cite{7917284,7827017,9354156,5594708}.
\textcolor{blue}{However, in multi-cell systems, since all BSs without data sharing, each UE is served by a dedicated BS and suffers from the co-channel interference caused by other UEs in the associated cell, and the cross-cell interference coming from adjacent cells also need to consider, especially for the UEs which are located close to the common boundary of multiple cells. 
Interference emerges as the one key capacity limiting factor in multi-cell systems. 
Unlike multi-cell systems, in the past decade, the cell-free networks, where there are no cells or cell boundaries and do not assign UEs to particular BSs, have been proposed as a potential technology to significantly improve the system performance with the multi-cell system \cite{7827017}.}  
Cell-free networks advocate a more active treatment of interference, where there are a large number of BSs serve multiple UEs in a large area of service cooperatively and simultaneously through exploiting a central processing unit (CPU) for optimizing the linear transmit beamforming or controlling the transmit powers \cite{7917284,9354156}.

Instead of cooperating through backhaul links, the relay-assisted communications can be beneficial not only in alleviating the channel fading effects but also in mitigating the interference \cite{5594708}.
Recently, intelligent reflecting surfaces (IRSs) have emerged as a promising technology for mitigating the detrimental propagation conditions and strengthening the desired signal powers through introducing additional links \cite{9140329,di2019smart,9110915}.  
IRSs are composed of a vast number of passive phase shifts, each of which can reconfigure the incident signals to desired directions with considerable array gains, to improve the system performance near-instantaneously. 
\textcolor{blue}{Besides, each phase shift can change the phase of the incident signal continuously or discretely.
Due to the hardware limitations, the discrete IRS model is more practical, where the phase shifts can only take finite discrete values from the discrete phase set and can be implemented by exploiting the positive-intrinsic-negative (PIN) diodes technique \cite{9279253,8930608,9352948}.}

However, the gains achieved by IRSs critically depend on the perfect channel state information (CSI), which is challenging to acquire due to IRSs are passive and can not sense channels. \textcolor{blue}{Some works designed the transmission protocol to separately estimate the IRS-related cascaded channels (i.e., BS-IRS-UE) and the direct channels (i.e., BS-UE) by performing an ON/OFF strategy \cite{9039554} of the phase shifts. 
The authors in \cite{9130088} proposed a novel three-phase pilot-based channel estimation framework for uplink multiple UEs communications, where the direct channels, the IRS-related cascaded channels of a typical UE, and the IRS-related cascaded channels of the other UEs are estimated in Phase \uppercase\expandafter{\romannumeral1}, \uppercase\expandafter{\romannumeral2}, and \uppercase\expandafter{\romannumeral3}, respectively.}
However, the implementing cost and computing complexity depend on the scale of IRSs, and it is very challenging to realize. 
To reduce the complexity, the authors in \cite{8937491} grouped the elements and then divided them into several sub-IRSs. 
Generally, the channels in IRS-assisted systems are estimated at the BSs through uplink pilots, and the so-obtained CSI are used to procode the transmitting data in the downlink. With the aid of a smart IRS controller, which can design the passive reflecting beamforming to achieve several objectives, e.g., transmitting power minimization \cite{9133435}, energy-efficient maximization \cite{9423652,8741198}, max-min fairness \cite{9246254}, and so on \cite{9120476,9316920,8955968}.

 Remarkably, the achievable sum-rate maximization problem is one of the key challenges faced by designers of deploying IRSs into communication systems. 
 The work \cite{8982186} considered the problem under MISO setting and proposed an efficient iterative alternating (IA) algorithm based on vector forms of the fractional programming methods \cite{shen2018fractional,shen2018fractional2}, i.e., Lagrangian dual transform (LDT) and quadratic transform (QT), and through joint optimizing the active transmit beamforming and the passive reflecting beamforming to maximize the sum-rate. 
 The authors in \cite{9394419} considered the multiple IRSs assisted full-duplex system, and they proposed a gradient ascent based resource allocation design algorithm to maximize the sum-rate of the uplink-downlink system. 
 \textcolor{blue}{The authors in \cite{9388932} investigated an IRS-assisted vehicular network and proposed a pair of algorithms to solve both the single UE case and the multiple UEs case, which were based on the successive refinement (SR) algorithm and alternating optimization method, respectively.
 The authors in \cite{9076830} maximized the sum-rate of all groups for IRS-assisted multiple groups multiple casts communications systems by proposing a majorization-minimization (MM) \cite{7547360} based algorithm. 
 The works \cite{9154244,9459505} investigated the problem in the cell-free scenario and consider a perfect IRS model, where both the amplitude and phase associated with each reflecting element of IRSs can be controlled independently and continuously. Besides, they exploited a similar algorithm as \cite{8982186} to solve the problem.}
The authors in \cite{9090356} studied the achievable sum-rate maximization problem in a multi-cell scenario, and they employed the weighted minimum mean-square error (WMMSE) \cite{5756489} technique to transform the original problem into an equivalent form and proposed an efficient algorithm by capitalizing on the block coordinate descent (BCD) and MM / complex circle manifold (CCM) methods to solve it. 
\textcolor{blue}{The authors in \cite{9279253} deployed the IRS to assist the joint processing coordinated multi-point (JP-CoMP) transmission. They studied two cases, i.e., the single UE case and the multiple UEs case. Particularly, the authors proposed a computational efficient MM-based algorithm to solve the problem in a single UE case and employed the second-order cone programming (SOCP) and semi-definite relaxation (SDR) \cite{6698378} techniques to solve the problem under multiple UEs case.} 

%\subsection{Contributions}
Motivated by the discussions as mentioned above, in this paper, we consider maximizing the achievable sum-rate of the distributed multiple IRSs assisted cell-free MIMO communication system. 
Since the optimizing variables are intricately coupled, and the constant modulus constraints introduced by IRSs, the formulated problem is non-convex and challenging to solve. To this end, we propose an efficient framework to jointly design the BSs, the UEs, and the IRSs. The main contributions of this paper are summarized as follows
\begin{itemize}
    \item As a compromise approach, we first transform the original non-convex problem to an equivalent form based on the fractional programming methods, i.e., the matrix forms of LDT and QT, and then decompose the reformulated problem into two subproblems, i.e., the active transmit beamforming matrices optimization and the passive reflecting beamforming matrices optimization. Consequently, we solve the subproblems alternating iteratively;
    \item For the active transmit beamforming matrices optimization subproblem, we reformulate it as a quadratically constrained quadratic programming (QCQP) problem. \textcolor{blue}{Based on the fact that the reformulated problem is convex and the dual gap is guaranteed to be zero, we propose a Lagrangian dual sub-gradient based algorithm to derive the locally optimal solutions of the active transmit beamforming matrices in nearly closed-forms};
    %\item \textcolor{red}{For the passive reflecting beamforming matrices optimization subproblem, we first transform the subproblem as a constant modulus constrained quadratic programming (QP) problem. We consider two cases of IRS model, i.e., the continuous phase case and the discrete phase case. For the continuous phase case, we propose an efficient algorithm to optimize the reflecting passive beamforming sequentially and obtain a feasible solution in a closed-form. 
    %While for the latter case, we propose an block-optimization algorithm to deal with the discrete restriction on the phase. We also prove the both proposed algorithms are guaranteed to converge to locally optimal solutions and with low-complexity;}
    \item For solving the passive reflecting beamforming matrices optimization subproblem, we first transform the problem into a tractable constant modulus constrained quadratic programming (CMC-QP) problem. \textcolor{blue}{Then we provide high-quality solutions by proposing a pair of relaxation-based algorithms, i.e., SDR and quadratic constraint relaxation (QCR) techniques. We also develop a low-complexity alternate sequential optimization (ASO) algorithm for solving the CMC-QP problem with the closed-form solutions. Particularly, we turn the CMC-QP problem into multiple one-dimensional optimization problems and solve them one-by-one while fixing the others. All three algorithms are guaranteed to converge to locally optimal solutions. 
    Meanwhile, the proposed ASO-based algorithm can extend to the discrete phase shift case after minor changes and employ an exhaustive-search strategy to solve it with low-complexity
    };
    \item Simulation results demonstrate the significant performance improvement achieved by deploying IRSs.
    \textcolor{blue}{Besides, the proposed algorithms outperform the benchmark schemes and have strong robustness to the CSI estimation errors}. 
\end{itemize}

\textit{Notations:} Vectors and matrices are presented by bold-face lower-case and upper-case letters, respectively. $\mathcal C\mathcal N\left(\mathbf 0, \mathbf I\right)$ denotes the circularly symmetric complex Gaussian (CSCG) distribution with zero mean and covariance matrix $\mathbf I$.  $\mathbf A^{\operatorname{H}}$, $\operatorname{Tr}\left(\mathbf A\right)$, and $\operatorname{Rank}\left(\mathbf A\right)$ denote the conjugate transpose, trace, and rank of the matrix $\mathbf A$, respectively. $\operatorname{Re}\left\{a\right\}$ is the real part of $a$.%$\mathbb C^{a\times b}$ denotes the space of $a\times b$ complex-value matrices a complex number.

\section{SYSTEM MODEL AND PROBLEM FORMULATION}
\begin{figure}
    \centering
    \includegraphics[width=0.5\linewidth]{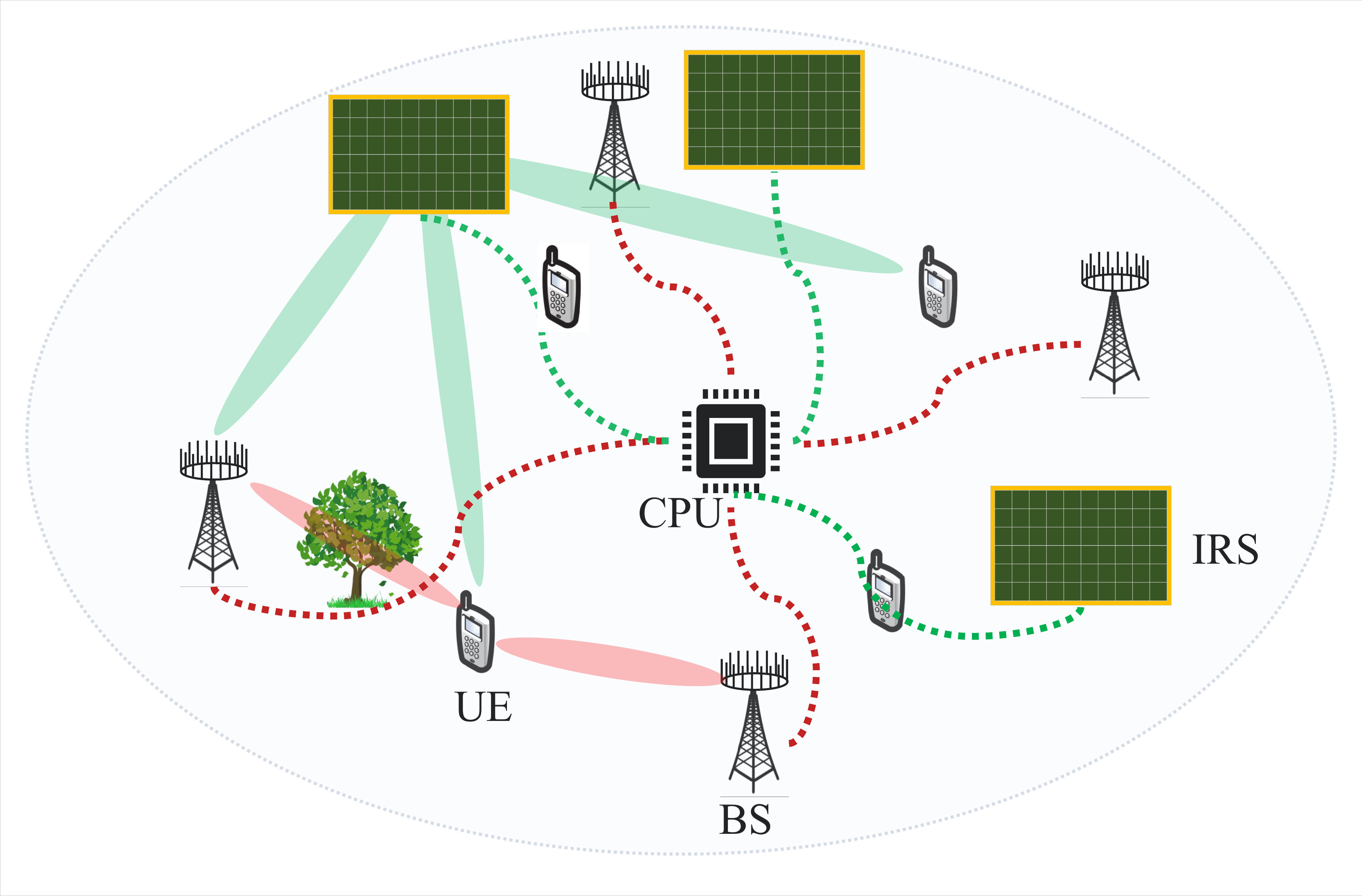}
    \caption{IRSs-assisted cell-free MIMO System Model.}
 \label{f1} 
\end{figure}
\textcolor{blue}{In this section, as shown in Fig. \ref{f1}, we describe distributed multiple IRSs assisted cell-free MIMO downlink communication systems, where $L$ distributed BSs and $R$ distributed IRSs are linked to a CPU and serve $K$ UEs cooperatively, and each BS, IRS, and UE are equipped with $M_b$ antennas, $N$ phase shifts and $M_u$ antennas, respectively.}  

\subsection{System Model}

The signal transmitted from the $L$ BSs can be mathematically expressed as
\begin{align}\label{eq1}
\mathbf{x} = \sum_{l=1}^L{\sum_{k = 1}^K \mathbf{W}_{l,k}\mathbf{s}_{l,k}}, \quad \forall\, l \in L,\forall k \in K,
\end{align}
where ${\mathbf s_{l,k}} \sim\mathcal C\mathcal N\left(\mathbf 0, \mathbf I_{M_u}\right)$ is the data symbol vector from the $l$-th BS intend for the $k$-th UE, and 
$\mathbf W _{l,k} \in {\mathbb{C}^{M_b \times M_u}}$ is the corresponding active transmit beamforming matrix from the $l$-th BS to the $k^{th}$ UE.
\textcolor{blue}{The received signal at the $k$-th UE is defined as
\begin{align}\label{eq2}
{{\mathbf y_{k}}} = \sum_{l=1}^L{\mathbf H _{l,k}^{\operatorname{H}} \mathbf W _{l,k} \mathbf s_{l,k}}+\sum_{l=1}^L\sum_{i=1,i\ne k}^K{\mathbf H _{l,k}^{\operatorname{H}} \mathbf W _{l,i} \mathbf s_{l,i}}+\mathbf n_k,\forall k \in K,
\end{align}
where $\mathbf n_{k} \sim  \mathcal{C}\mathcal{N}\left( {\mathbf 0,{\sigma ^2\mathbf {I}_{M_u}}} \right)$ is the background additive white Gaussian noise (AWGN) vector at the $k$-th UE, and $\mathbf H_{l,k} \in \mathbb{C}^{M_b \times M_u}$ denotes the equivalent channel spanning from the $l$-th BS to the $k$-th UE, each of which can be expressed as
\begin{align}\label{eq3}
{\mathbf H_{l,k}^{\operatorname{H}} = \mathbf D_{l,k}^{\operatorname{H}} + \sum_{r=1}^R{\mathbf G_{r,k}^{\operatorname{H}}\Theta_r {\mathbf S_{l,r}}}},
\end{align}
where ${\mathbf D_{l,k} \in {\mathbb{C}^{M_b \times M_u}}}$ denotes the direct channel from the $l$-th BS to the $k$-th UE, and ${\mathbf G_{r,k} \in \mathbb{C}^{N \times M_u}}$ is the channel from the $r$-th IRS to the $k$-th UE. The channel from the $l$-th BS to the $r$-th IRS is represented by ${\mathbf S_{l,r} \in {\mathbb{C}^{N \times M_b}}}$. With the IRSs, the transmitted signal can be dynamically altered by its $N$ phase shifts, where the passive reflecting beamforming matrix are denoted as
\begin{align}\label{eq4}
    \Theta_r=\alpha\operatorname{diag}\left(e^{j{\phi_{r,1}}},e^{j{\phi_{r,2}}},\cdots,e^{j{\phi_{r,N}}}\right),\forall r \in R.
\end{align}
where $\alpha \in \left[ {0,1} \right]$ denotes the reflecting efficiency of IRSs (the signal power loss is caused by the signal absorption of the phase shifts), and  $\phi_{r,n} \in \left[0,2\pi\right)$ is the phase of the $n$-th phase shift of the $r$-th IRS.} Meanwhile, we also investigate the discrete phase shift case in Sec.\ref{discretecase}.

\subsection{Channel Model}
We assume the small-scale fading of the direct channels (i.e., $\mathbf D_{l,k},\forall\left\{{l,k}\right\}$) and the IRS-related channels (i.e., $\mathbf G_{r,k},\forall\left\{{r,k}\right\}$ and $\mathbf S_{l,r},\forall \left\{l,r\right\}$) are Rayleigh fading and Rician fading \cite{9090356,9279253}, respectively, which are modeled as 
\begin{align}
    \mathbf D_{l,k}&=\mathcal P_{d_{lk}}\left(d_{lk}\right)\tilde{\mathbf {D}}_{l,k},\quad \mathbf G_{r,k}=\mathcal P_{d_{rk}}\left(d_{rk}\right)\left(\sqrt{\frac{\beta_G}{1+\beta_G}}\bar{\mathbf {G}}_{r,k}+\sqrt{\frac{1}{1+\beta_G}}\tilde{\mathbf {G}}_{r,k}\right),\notag\\
    \mathbf S_{l,r}&=\mathcal P_{d_{lr}}\left(d_{lr}\right)\left(\sqrt{\frac{\beta_S}{1+\beta_S}}\bar{\mathbf {S}}_{l,r}+\sqrt{\frac{1}{1+\beta_S}}\tilde{\mathbf {S}}_{l,r}\right),
\end{align}
where $\bar{\mathbf {G}}_{r,k}$ ($\bar{\mathbf {S}}_{l,r}$) and $\tilde{\mathbf {G}}_{r,k}$ ($\tilde{\mathbf {D}}_{l,k},\tilde{\mathbf {S}}_{l,r}$) are the line-of-sight (LOS) and non-LOS (NLOS) components of the channels ${\mathbf {G}}_{r,k}$ (${\mathbf {D}}_{l,k},{\mathbf {S}}_{l,r}$), and $\mathcal P_{d_{i}}\left(d_{i}\right), \forall i \in \left\{lk,rk,lr\right\}$ are the distance-dependent large-scale path loss, which are defined as 
\begin{align}
    \mathcal P_{d_{i}}\left(d_{i}\right)= \mathcal C_0 \left(\frac{d_i}{d_0}\right)^{-p_i},\forall i \in \left\{lk,rk,lr\right\},
\end{align}
where $\mathcal C_0$ is the path loss at the reference distance $d_0=1$m, and $d_i$ and $p_i$ denote the distances and path loss exponents between $l$-th BS with $k$-th UE link (between $r$-th IRS with $k$-th UE link, between $l$-th BS with $r$-th IRS link), respectively. 

\textcolor{blue}{Meanwhile, we further assume that the BSs and the UEs are equipped with uniform linear arrays (ULAs), and the IRSs are modeled as uniform planar arrays (UPAs) \cite{9133435,9423652,9279253}.  We set $\bar d={\bar\omega}/{2}$, where $\bar d$ and $\bar\omega$ denote antenna spacing and wavelength, respectively. The departure (arrival) antenna steering vectors at the $l$-th BS ($k$-th UE) are given by 
\begin{align}
    {\mathbf a_{l}}\left( {{\vartheta_l ^{AoD}}} \right) &= {\left[ {1,{e^{j{{\pi}}\sin {\vartheta_l ^{AoD}}}},\cdots,{e^{j{{\pi }}\left( {{M_b} - 1} \right)\sin {\vartheta_l ^{AoD}}}}} \right]^{\operatorname{T}}},\notag\\
    {\mathbf a_{k}}\left( {{\vartheta_k ^{AoA}}} \right) &= {\left[ {1,{e^{j{{\pi }}\sin {\vartheta_k ^{AoA}}}},\cdots,{e^{j{{\pi }}\left( {{M_u} - 1} \right)\sin {\vartheta_k ^{AoA}}}}} \right]^{\operatorname{T}}},
\end{align}  
where $\vartheta_l ^{AoD}$ and $\vartheta_k ^{AoA}$ are the angles of departure and arrival, respectively. The antenna steering vector at $r$-th IRS is modeled as 
\begin{align}
    {\mathbf  a}_{r}\left(\kappa_r,\varphi_r\right)={\mathbf  a}_{r}^v\left(\kappa_r,\varphi_r\right) \otimes {\mathbf  a}_{r}^h\left(\varphi_r\right),
\end{align}
where 
\begin{align}
    {\mathbf  a}_{r}^v\left(\kappa_r,\varphi_r\right)&={\left[ {1,{e^{j{{\pi}}\sin {\kappa_r}\sin {\varphi_r}}},\cdots,{e^{j{{\pi}}\left( {{N_v} - 1} \right)\sin {\kappa_r}\sin {\varphi_r}}}} \right]^{\operatorname{T}}},\notag\\
    {\mathbf  a}_{r}^h\left(\varphi_r\right)&={\left[ {1,{e^{j{{\pi}}\cos {\varphi_r}}},\cdots,{e^{j{{\pi}}\left( {{N_h} - 1} \right)\cos {\varphi_r}}}} \right]^{\operatorname{T}}},
\end{align}
where ${\kappa_r}$ and $\varphi_r$ are the azimuth and elevation angle, respectively, and where $N_v$ and $N_h$ are the number of phase shifts along the vertical and horizontal directions, respectively, and satisfy $N=N_v N_h$.
Therefore, the LOS components of ${\mathbf {G}}_{r,k}$ and ${\mathbf {S}}_{l,r}$, i.e., $\bar{\mathbf {G}}_{r,k}$ and $\bar{\mathbf {S}}_{l,r}$ are defined as
\begin{align}
    \bar{\mathbf {G}}_{r,k}={\mathbf a_{k}}\left( {{\vartheta_k ^{AoA}}} \right)\hat{\mathbf  a}_{r}^{\operatorname{H}}\left(\kappa_r^{AoD},\varphi_r^{AoD}\right),\quad
    \bar{\mathbf {S}}_{l,r}=\bar{\mathbf  a}_{r}\left(\kappa_r^{AoA},\varphi_r^{AoA}\right){\mathbf a_{l}^{\operatorname{H}}}\left( {{\vartheta_l ^{AoD}}} \right),
\end{align}
where $\hat{\mathbf  a}_{r}\left(\kappa_r^{AoD},\varphi_r^{AoD}\right)$ and $\bar{\mathbf  a}_{r}\left(\kappa_r^{AoA},\varphi_r^{AoA}\right)$ are departure and arrival antenna steering vectors at $r$-th IRS, respectively.}

\begin{figure}
    \centering
    \includegraphics[width=0.5\linewidth]{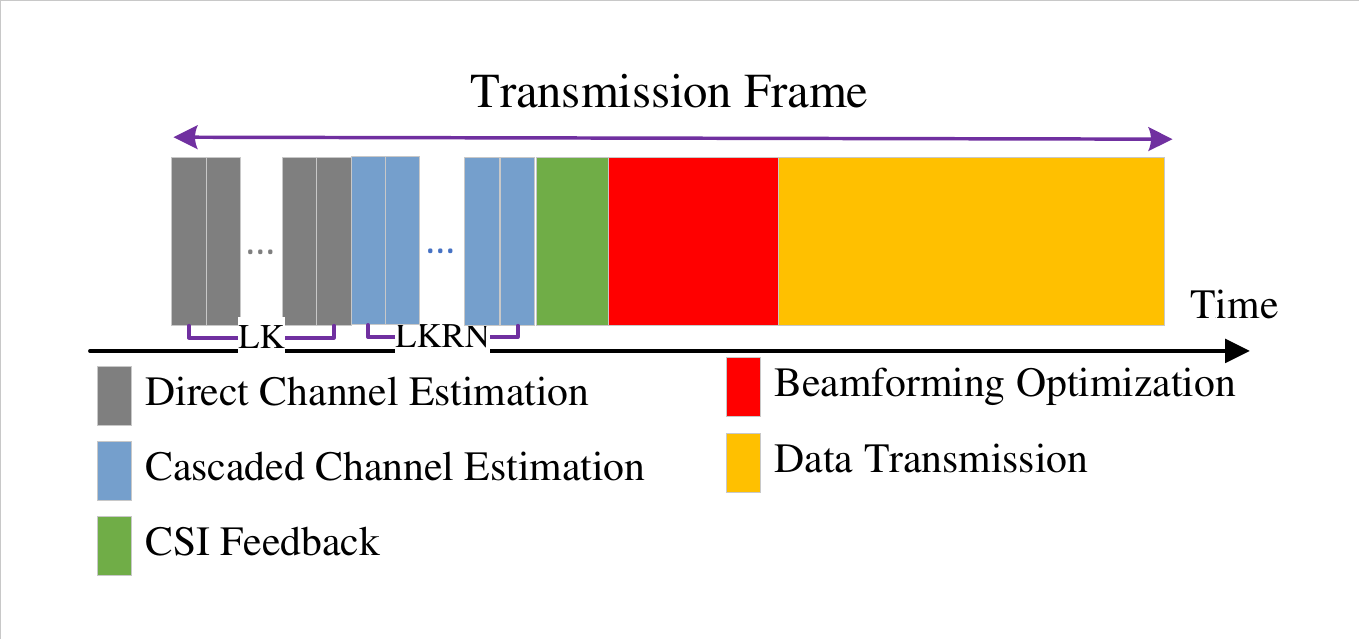}
    \caption{Transmission protocol.}
 \label{ftransmission} 
\end{figure}

As shown in Fig. \ref{ftransmission}, in this paper, \textcolor{blue}{to characterize the theoretical limit of the considered cell-free system, we assume that the full CSI knowledge on all links involved (includes $LK$ direct channels and $LKRN$ IRS-related cascaded channels) are available at the BSs perfectly by using the customized CSI estimation schemes\footnote{
In Sec.\ref{imcsi}, we also investigate the robustness of the proposed jointly design framework to the CSI estimation errors. However, the robust transmission design scheme \cite{9117093,9180053,5982443} with inaccurate CSI needs to be proposed, which will be left as future work.}, and the so-obtained CSI can be feedback to the CPU.} Then the CPU uses the CSI to design the corresponding active transmit beamforming matrices and the passive reflecting beamforming matrices. Finally, the BSs and the IRSs exploit the beamforming matrices to precode and reconfigure the signals. 

\subsection{Problem Formulation}

Based on the aforementioned discussions, the signal to interference plus noise ratio (SINR) matrix at the $k$-th UE can be formulated as\textcolor{blue}{ 
\begin{align}\label{eq5}
    {\Gamma _{k}} =\sum_{l=1}^L{ \mathbf H_{l,k}^{\operatorname{H}}{\mathbf W_{l,k}}\mathbf W_{l,k}^{\operatorname{H}}{\mathbf H_{l,k}}}\mathbf V_k^{\operatorname{-1}},\forall k \in K,
\end{align}}
where \textcolor{blue}{
\begin{align}\label{eq6}
  \mathbf V_{k} = &\sum_{l = 1}^L\sum_{i=1,i\ne k}^K{\mathbf H_{l,k}^{\operatorname{H}}{\mathbf W_{l,i}}\mathbf W_{l,i}^{\operatorname{H}}{\mathbf H_{l,k}}}  + {\sigma ^2}{\mathbf {I}_{{M_u}}}. 
\end{align}}
In this paper, our objective is maximizing the achievable sum-rate of all UEs, i.e., $\mathcal{R}\left( {\mathbf W,\Theta} \right) = {\sum_{k = 1}^K {\log \left| {\mathbf I + {\Gamma _{k}}} \right|} }$, through jointly optimizing the active transmit beamforming matrices $\mathbf W=\left\{\mathbf W_{l,k},\forall {l,k}\right\}$ at $L$ BSs and the passive reflecting beamforming matrices $\Theta=\left\{\Theta_r,\forall r\right\}$ at $R$ IRSs. The achievable sum-rate maximization problem is formulated as
\begin{align}
{ \mathop {\max }_{\mathbf W ,\Theta } } \quad & {{\mathcal{R}}\left( {\mathbf W,\Theta } \right)} \label{eq7}\\
{\operatorname{s.t.} \quad} & {\sum_{k = 1}^K {{{\left\| {{\mathbf W _{l,k}}} \right\|}_F ^2}}  \le {P_{\max,l }},\forall\, l \in L}\tag{13a}\label{eq8},\\
&  \textcolor{blue}{\left  |\Theta_{r,n} \right |={\alpha},\forall r \in R, \forall n \in N \tag{13b}\label{eq9}}, 
\end{align} 
where constrains \eqref{eq8} limit the maximum transmit power of each BS and \eqref{eq9} represent the constant modulus constraint of each phase shift at the IRSs.

\section{Proposed Joint Optimization Algorithm}
The formulated Problem \ref{eq7} turns out to be non-convex and not tractable. \textcolor{blue}{As a compromise approach, we consider transforming the problem to an equivalent form by exploiting the fractional programming methods in matrix-forms, e.g., QT 
\cite[Corollary 1]{shen2018fractional2} and LDT \cite[Theorem 4]{shen2018fractional2}}.

First, we consider moving the matrix-ratio terms out of the logarithm function in $\mathcal R \left(\mathbf W,\Theta\right)$ by employing LDT method as follows.

\begin{proposition}[Matrix-form of LDT]\label{p1}
By introducing the auxiliary diagonal matrices $\mathbf U_k \in \mathbb{C} ^{M_u\times M_u}, \forall k \in K$, the sum-rate maximization Problem \ref{eq7} can be equivalently transformed as 
\begin{align}\label{eq8pro1}
{ \mathop {\max }_{\mathbf W ,\Theta, \mathbf U } } \quad & {{f_1}\left( {\mathbf W,\Theta,\mathbf U } \right)}\\
{\operatorname{s.t.} \quad} & {\eqref{eq8},\eqref{eq9}},\notag
\end{align}
where the new objective function ${f_1}\left( {\mathbf W,\Theta,\mathbf U } \right)$ is defined as  
\begin{align}\label{eq90}
{f_1}\left( {\mathbf W,\Theta ,\mathbf U } \right)  = \sum_{k=1}^K{\log\left|\mathbf I _{M_u}+\mathbf U_k\right|}-\sum_{k=1}^K{\operatorname{Tr}\left(\mathbf U_k\right)} +\sum_{k=1}^K{\operatorname{Tr}\left(\left(\mathbf I _{M_u}+\mathbf U_k\right)f_2\left(\mathbf W,\Theta\right)\right)},
\end{align}
where $f_2\left ( \mathbf W, \Theta \right )$ included in the function is given as
\begin{align}\label{eq10}
    {f_2}\left( {\mathbf W,\Theta } \right) = \sum_{l=1}^L{\mathbf H_{l,k}^{\operatorname{H}}{\mathbf W_{l,k}} \mathbf W_{l,k}^{\operatorname{H}}{\mathbf H_{l,k}}}\mathbf{\bar{V}}_{k}^{-1},
\end{align}
where
\begin{align}\label{eq11}
    \mathbf{\bar V}_{k} = \sum_{l = 1}^L {\sum_{i = 1}^K {\mathbf H_{l,k}^{\operatorname{H}}{\mathbf W_{l,i}}\mathbf W_{l,i}^{\operatorname{H}}{\mathbf H_{l,k}}} }  + {\sigma ^2}{\mathbf I_{{M_u}}}.
\end{align}
\end{proposition}

\begin{proof}
The proof is presented in Appendix \ref{ap1}.
\end{proof} 

Since the optimization variables are still intricately coupled in the matrix-ratio terms, Problem \ref{eq8pro1} is still non-convex and challenging to handle directly. Based on the QT method, we have the following proposition.

\begin{proposition}[Matrix Form of QT]\label{p2}
By introducing the auxiliary diagonal matrices $\mathbf Y_k\in \mathbb{C}^{M_u \times M_u}$, $\forall k$, Problem \ref{eq8pro1} can be equivalently transformed as
\begin{align}\label{eq13}
    \mathop{\max}_{\mathbf W, \Theta, \mathbf U, \mathbf Y} \quad & f_3\left(\mathbf W, \Theta, \mathbf U, \mathbf Y\right)\\
    {\operatorname{s.t.}}\quad &\eqref{eq8},\eqref{eq9},\notag
\end{align}
where $f_3\left(\mathbf W, \Theta, \mathbf U, \mathbf Y\right)$ is formulated as 
\begin{align}
    &f_3\left(\mathbf W, \Theta, \mathbf U, \mathbf Y\right)=\sum_{k=1}^K{\log\left|\mathbf {\bar U}_k\right|}-\sum_{k=1}^K{\operatorname{Tr}\left(\mathbf U_k\right)}+\sum_{k=1}^K{\operatorname{Tr}\left(\mathbf {\bar U}_k \mathbf Y_k^{\operatorname{H}}\sum_{l=1}^L{\mathbf H_{l,k}^{\operatorname{H}}\mathbf W_{l,k}}\right)}\notag\\
    &+\sum_{k=1}^K{\operatorname{Tr}\left(\mathbf {\bar U}_k \sum_{l=1}^L{\mathbf W_{l,k}^{\operatorname{H}}\mathbf H_{l,k}}\mathbf Y_k\right)}-\sum_{k=1}^K{\operatorname{Tr}\left({\mathbf{\bar U}_k}\mathbf Y_k^{\operatorname{H}}\sum_{l=1}^L\sum_{i=1}^K{\mathbf H_{l,k}^{\operatorname{H}}\mathbf W_{l,i}\mathbf W_{l,i}^{\operatorname{H}}\mathbf H_{l,k}}\mathbf Y_{k}\right)}-\sum_{k=1}^K{\operatorname{Tr}\left(\sigma_k^2{\mathbf{\bar U}_k}{\mathbf Y_k^{\operatorname{H}}}\mathbf Y_{k}\right)}.
\end{align}
where $\mathbf {\bar U}_k \triangleq {\mathbf {U} _{k}}+\mathbf I_{M_u},\forall k \in K$.
\end{proposition}
\begin{proof}\textcolor{blue}{
This proposition is extended QT method from the vector-form \cite[Corollary 1]{shen2018fractional2} to the matrix-form and can be easily proved as follows. 
Note that $f_3\left(\mathbf W, \Theta, \mathbf U, \mathbf Y\right)$ in Problem \ref{eq13} is convex with respect to $\mathbf Y_k,\forall k$. Thus, by setting the partial derivative of this objective function with respect to $\mathbf Y_k,\forall k$ to be zeros, we can obtain the optimal solutions in closed-forms, which is expressed as \eqref{eq18}.
Substituting the optimal $\mathbf Y_k,\forall k$ back into the objective function of Problem \ref{eq13} recovers Problem \ref{eq8pro1} equivalently. 
The proof is completed.}

\textcolor{blue}{Meanwhile, note that $\mathbf Y_k,\forall k$ can be treated as the decoding matrix to decode the received signals at the UEs, which is the same as the minimum mean-square error (MMSE) receive filter \cite[Sec.VI-A]{shen2018fractional2}. For more information of QT method, please refer to \cite{shen2018fractional,shen2018fractional2}.}
\end{proof}

Although the problem has been significantly simplified, Problem \ref{eq13} is still jointly non-convex. Fortunately,  $f_3\left(\mathbf W, \Theta, \mathbf U, \mathbf Y\right)$ is convex with respect to any one of the four variables $\mathbf U$, $\mathbf Y$, $\mathbf W$, and $\Theta$ when the other three are fixed. Therefore, in the following subsections, we decompose Problem \ref{eq13} into several subproblems and solve them alternating iteratively. For notation convenience, the solutions after the $t$-th iteration are denoted by $\left(\cdot\right)^{\left(\operatorname{t+1}\right)}$.

\subsection{Optimization of SINR}
\textcolor{blue}{First, we consider optimizing $\mathbf U$, i.e., SINR.} With fixed $\mathbf W^{\left(\operatorname{t}\right)}$, $\Theta^{\left(\operatorname{t}\right)}$ and $\mathbf Y^{\left(\operatorname{t}\right)}$, the optimal solution of $\mathbf U^{\left(\operatorname{t+1}\right)}$ can be obtained by solving the following subproblem
\begin{align}\label{eq14}
    \mathbf U^{\left(\operatorname{t+1}\right)} \triangleq \arg\mathop{\max} _\mathbf U {f_3}\left( {\mathbf W^{\left(\operatorname{t}\right)},\Theta^{\left(\operatorname{t}\right)} ,\mathbf U,\mathbf Y^{\left(\operatorname{t}\right)} } \right).
\end{align}
Note that the auxiliary matrices $\mathbf U_k,\forall k$ only appear in the objective function of subproblem \ref{eq14} and do not exist in any constraints. Therefore, by setting the partial derivatives of the objective function in subproblem \ref{eq14} with respect to $\mathbf U_k,\forall k$ to be zeros, and after some matrix manipulations, the closed-form solutions are optimal given as
\begin{align}\label{eq15}
    \mathbf U_k^{\left(\operatorname{t+1}\right)}= \Gamma _k, \forall k.
\end{align}
Note that the variables $\mathbf W$, $\Theta$ and $\mathbf Y$ only exist in partial terms of ${f_3}\left( {\mathbf W,\Theta ,\mathbf U,\mathbf Y} \right)$, therefore, we recast the objective function as $f_3\left( {\mathbf W,\Theta ,\mathbf U,\mathbf Y} \right)\triangleq f_4\left( {\mathbf W,\Theta ,\mathbf Y} \right)+\operatorname{Const}\left(\mathbf U\right)$, where 
\begin{align}\label{eq16}
    f_4\left(\mathbf W, \Theta, \mathbf Y\right)&=
    \sum_{k=1}^K{\operatorname{Tr}\left(\mathbf {\bar U}_k\mathbf Y_k^{\operatorname{H}} \sum_{l=1}^L{\mathbf H_{l,k}^{\operatorname{H}}\mathbf W_{l,k}}\right)}+\sum_{k=1}^K{\operatorname{Tr}\left(\mathbf {\bar U}_k \sum_{l=1}^L{\mathbf W_{l,k}^{\operatorname{H}}\mathbf H_{l,k}}\mathbf Y_k\right)}\notag\\
    &-\sum_{k=1}^K{\operatorname{Tr}\left({\mathbf{\bar U}_k}{\mathbf Y_k^{\operatorname{H}}}\sum_{l=1}^L\sum_{i=1}^K{\mathbf H_{l,k}^{\operatorname{H}}\mathbf W_{l,i}\mathbf W_{l,i}^{\operatorname{H}}\mathbf H_{l,k}}\mathbf Y_{k}\right)}-\sum_{k=1}^K{\operatorname{Tr}\left(\sigma_k^2{\mathbf{\bar U}_k}{\mathbf Y_k^{\operatorname{H}}}\mathbf Y_{k}\right)},
\end{align}
and $\operatorname{Const}\left(\mathbf U\right)=\sum_{k=1}^K{\log\left|\mathbf {\bar U}_k\right|-\sum_{k=1}^K{\operatorname{Tr}\left(\mathbf U_k\right)}}$, which has no impact on the optimization of $\mathbf W$, $\Theta$ and $\mathbf Y$.
Therefore, with the fixed variable of $\mathbf U_k^{\operatorname{\left(t+1\right)}}, \forall k$, we can optimize other variables by only investigating $f_4\left(\mathbf W, \Theta, \mathbf Y\right)$.

\subsection{Optimization of Decoding Matrix}
Now, we optimize the variables $\mathbf Y_k, \forall k$. With the fixed variables of $\mathbf W^{\left(\operatorname{t}\right)}$, $\Theta^{\left(\operatorname{t}\right)}$, and $\mathbf U^{\left(\operatorname{t+1}\right)}$, the subproblem of optimization $\mathbf Y$ is given as
\begin{align}
    \mathbf Y^{\left(\operatorname{t+1}\right)} \triangleq \arg\mathop{\max} _\mathbf Y {f_4}\left( {\mathbf W^{\left(\operatorname{t}\right)},\Theta^{\left(\operatorname{t}\right)},\mathbf Y } \right).
\end{align}
\textcolor{blue}{Similar to the method for solving the subproblem of $\mathbf U$, the variable $\mathbf Y$ does not exist in any constraint set, hence, by setting the partial derivatives of ${f_4}\left( {\mathbf W^{\left(\operatorname{t}\right)},\Theta^{\left(\operatorname{t}\right)},\mathbf Y } \right)$ with regard to $\mathbf Y_k, \forall k$ to be zeros, the closed-form solutions of $\mathbf Y_k,\forall k$ are optimal expressed as
\begin{align}\label{eq18}
    \mathbf Y_k^{\left(\operatorname{t+1}\right)}=\sum_{l=1}^L{\mathbf H_{l,k}^{\operatorname{H}}\mathbf W_{l,k}}\mathbf {\bar V}_k^{\operatorname{-1}}, \forall k.
\end{align}}

\subsection{Optimization of Active Transmit Beamforming}
At the BSs side, with the fixed $\Theta^{\left(\operatorname{t}\right)}$, and the obtained solutions of variables $\mathbf U^{\left(\operatorname{t+1}\right)}$ and $\mathbf Y^{\left(\operatorname{t+1}\right)}$, the optimization subproblem of the active transmit beamforming $\mathbf W$ is defined as
\begin{align}
    \mathbf W^{\left(\operatorname{t+1}\right)} \triangleq \left\{\arg\mathop{\max} _\mathbf W \; {f_4}\left( {\mathbf W,\Theta^{\left(\operatorname{t}\right)} , \mathbf Y^{\left(\operatorname{t+1}\right)} } \right),\quad \operatorname{s.t.}  \eqref{eq8}\right\},
\end{align}
By omitting the irrelevant constant term with respect to $\mathbf W$, i.e., $\sum_{k=1}^K{\operatorname{Tr}\left(\sigma_k^2{\mathbf{\bar U}_k}{\mathbf Y_k^{\operatorname{H}}}\mathbf Y_{k}\right)}$, the above problem can be equivalently transformed as
\begin{align}\label{eq21}
   \mathbf W^{\left(\operatorname{t+1}\right)} \triangleq \left\{\arg \mathop{\min} _\mathbf W \; {f_5}\left( \mathbf W\right),\quad
    \operatorname{s.t.} \; \eqref{eq8}\right\},
\end{align}
where
\begin{align}
    f_5\left(\mathbf W\right)
    &=\sum_{k=1}^K{\operatorname{Tr}\left({\mathbf{\bar U}_k}{\mathbf Y_k^{\operatorname{H}}}\sum_{l=1}^L\sum_{i=1}^K{\mathbf H_{l,k}^{\operatorname{H}}\mathbf W_{l,i}\mathbf W_{l,i}^{\operatorname{H}}\mathbf H_{l,k}}\mathbf Y_{k}\right)}\notag\\
    &-\sum_{k=1}^K{\operatorname{Tr}\left(\mathbf {\bar U}_k\mathbf Y_k^{\operatorname{H}} \sum_{l=1}^L{\mathbf H_{l,k}^{\operatorname{H}}\mathbf W_{l,k}}\right)}-\sum_{k=1}^K{\operatorname{Tr}\left(\mathbf {\bar U}_k \sum_{l=1}^L{\mathbf W_{l,k}^{\operatorname{H}}\mathbf H_{l,k}}\mathbf Y_k\right)}.
\end{align}
It can be verified that the objective function $f_5\left(\mathbf W\right)$ and the constraints \eqref{eq8} of Problem \ref{eq21} are both convex with respect to $\mathbf W$, therefore, Problem \ref{eq21} is a QCQP problem, which can be solved by employing convex solver tools, e.g., CVX \cite{cvx}. \textcolor{blue}{Instead of relying on the generic solver with high computational complexity, in the following, we propose a Lagrangian dual sub-gradient \cite{boyd2004convex} based algorithm to optimize $\mathbf W$.}

\textcolor{blue}{The Lagrangian dual function of Problem \ref{eq21} is defining as
\begin{align}
    \mathcal L\left(\mathbf W,\lambda\right)=f_5\left(\mathbf W\right)
    +\sum_{l=1}^L{\lambda_l\left(\sum_{i=1}^K{\operatorname{Tr}\left(\mathbf W_{l,i}^{\operatorname{H}}\mathbf W_{l,i}\right)}-P_{\max,l}\right)},
\end{align}
where $\mathbf \lambda _l \ge 0$ is the dual variable introduced for enforcing the maximal power constraint in $l$-th BS.  
The variables of $\mathbf W$ and $\lambda$ can be obtained alternately. The solutions of $\left\{\mathbf W, \lambda\right\}$ in $q$-th sub-iteration can be denoted by $\mathbf W^{\operatorname{q+1}}$ and $\lambda^{\operatorname{q+1}}$. 
First, with fixed $\lambda^{\operatorname{q}}$, the optimal active transmit beamforming can be optimized by solving the following problem 
\begin{align}
    \mathbf W^{\operatorname{q+1}}&=\arg\mathop{\min}_{\mathbf W} \mathcal L\left(\mathbf W,\lambda^{\operatorname{q}}\right),
\end{align}
By setting the first-order partial derivative of $\mathcal L\left(\mathbf W,\lambda^{\operatorname{q}}\right)$ with respect to $\mathbf W_{l,k},\forall \left\{l,k\right\}$ to be zeros, we have
\begin{align}\label{eq24}
    \mathbf W_{l,k}^{\operatorname{q+1}}\left(\lambda_l^{\operatorname{q}}\right)=\frac{\mathbf H_{l,k}\mathbf Y_k\mathbf {\bar U}_k}{\sum_{i=1}^K{\operatorname{Tr}\left(\mathbf {\bar U}_i \mathbf Y_i^{\operatorname{H}} \mathbf H_{l,i}^{\operatorname{H}}\mathbf H_{l,i}\mathbf Y_i\right)}+\lambda_l^{\operatorname{q}}\mathbf I_{M_u}}, \forall \left\{l,k\right\}.
\end{align}}
\textcolor{blue}{
The dual variable $\lambda$ can be determined by solving the following dual optimizing problem, which is given as
\begin{align}\label{eq25}
    \lambda^{\operatorname{q+1}}&=\arg\mathop{\max}_{\lambda_l\ge0,\forall l}  g\left(\lambda\right),
\end{align}
where 
\begin{align}
    g\left(\lambda\right)=\mathop{\min}_{\mathbf W} \mathcal L\left(\mathbf W^{\operatorname{q+1}},\lambda\right).
\end{align} 
By defining $f_l\left(\lambda_l^{\operatorname{q}}\right)=\sum_{k=1}^K{\operatorname{Tr}\left(\left(\mathbf {W}_{l,k}^{\operatorname{q+1}}\left(\lambda_l^{\operatorname{q}}\right)\right)^{\operatorname{H}}  \mathbf {W}_{l,k}^{\operatorname{q+1}}\left(\lambda_l^{\operatorname{q}}\right) \right)}-P_{\max,l}$, which is a monotonically decreasing function for $\lambda_l\ge0$, we propose a sub-gradient based method to update $\lambda_l^{\operatorname{q+1}}$. Particularly,  
with the fixed $\mathbf W_{l,k}^{\operatorname{q+1}}$, the dual variable $\lambda_l$ can be performed as follows
\begin{align}\label{eq27}
    \lambda_l^{\operatorname{q+1}}=\left[\lambda_l^{\operatorname{q}}+\tau_l f_l\left(\lambda_l^{\operatorname{q}}\right)\right]^+,\forall l.
\end{align}
where $\tau_l$ denotes the positive step size for updating $\lambda_l^{\operatorname{q+1}}$ and $\left[x\right]^+=\max\left\{x,0\right\}$.} The detailed information of the Lagrangian dual sub-gradient based algorithm for optimizing the active transmit beamforming are summarized in Algorithm \ref{a1}. % We have $\mathbf W^{\operatorname{\left(t+1\right)}}\triangleq\mathbf W^{\operatorname{q_{\max}}}$, where $\operatorname{q}_{\max}$ denotes the maximum number of sub-iterations when Algorithm \ref{a1} converges.
\begin{algorithm}[t]
\caption{Lagrangian Dual Sub-gradient based Algorithm for Optimizing Problem \ref{eq21}} 
\label{a1} 
\begin{algorithmic}[1]
\REQUIRE dual variables $\left\{\lambda_l^{\operatorname{1}}\ge0,\forall l\right\}$, step size $\left\{\tau_l\ge0,\forall l\right\}$, threshold $\varepsilon_1$, $q=1$.
\STATE \textbf{Update} $\mathbf W_{l,k} ^{\operatorname{q+1}}, \forall \left\{l,k\right\}$ using \eqref{eq24};
\STATE \textbf{Update $\lambda_l^{\operatorname{q+1}},\forall l$} using \eqref{eq27};
\STATE \textbf{If} {${\left|\lambda_l^{\operatorname{q+1}} - \lambda_l^{\operatorname{q}} \right|}/{\lambda_l^{\operatorname{q+1}}} < \varepsilon_1 $}, output $\mathbf W_{l,k}^{\left(\operatorname{t+1}\right)}\triangleq  \mathbf W_{l,k}^{\operatorname{q}},\forall \{l,k\}$;\\\textbf{Otherwise}, set $q=q+1$ and go to step 1.
\end{algorithmic} 
\end{algorithm}

\subsection{Optimization of Passive Reflecting Beamforming}
Finally, we consider to optimize $\Theta$. With fixed $\mathbf U^{\left(\operatorname{t+1}\right)}$, $\mathbf W^{\left(\operatorname{t+1}\right)}$ and $\mathbf Y^{\left(\operatorname{t+1}\right)}$, the corresponding subproblem for optimizing the passive reflecting beamforming matrices of IRSs are given as
\begin{align}
    \Theta^{\left(\operatorname{t+1}\right)} \triangleq \left\{\arg\mathop{\max} _\Theta \;{f_4}\left( {\mathbf W^{\left(\operatorname{t+1}\right)},\Theta,\mathbf Y^{\left(\operatorname{t+1}\right)} } \right),\quad
    \operatorname{s.t.} \;\eqref{eq9}\right\}.
\end{align}
\textcolor{blue}{By defining $\Theta=\operatorname{diag}\left\{\Theta_1,\Theta_2,\cdots,\Theta_R\right\}\in \mathbb C^{RN\times RN}$, $\mathbf W_i=\left[\mathbf W_{1,i}^{\operatorname{T}},\mathbf W_{2,i}^{\operatorname{T}},\cdots,\mathbf W_{L,i}^{\operatorname{T}}\right]^{\operatorname{T}}\in \mathbb C^{L M_B\times M_U}$, $\mathbf D_{k}=\left[\mathbf D_{1,k}^{\operatorname{T}},\mathbf D_{2,k}^{\operatorname{T}},\cdots,\mathbf D_{L,k}^{\operatorname{T}}\right]^{\operatorname{T}}\in \mathbb C^{L M_B\times M_U}$, $\mathbf G_{k}=\left[\mathbf G_{1,k}^{\operatorname{T}},\mathbf G_{2,k}^{\operatorname{T}},\cdots,\mathbf G_{R,k}^{\operatorname{T}}\right]^{\operatorname{T}}\in \mathbb C^{RN\times M_U}$, $\mathbf S_{r}=\left[\mathbf S_{1,r},\mathbf S_{2,r},\cdots,\mathbf S_{L,r}\right]\in \mathbb C^{N\times L M_B}$, 
\textcolor{red}{\begin{align}\label{eq29}
    \mathbf S=\left[\begin{matrix}
    \mathbf S_{1,1}^{\operatorname{T}}&\mathbf S_{1,2}^{\operatorname{T}}&\cdots&\mathbf S_{1,R}^{\operatorname{T}}\\
    \mathbf S_{2,1}^{\operatorname{T}}&\mathbf S_{2,2}^{\operatorname{T}}&\cdots&\mathbf S_{2,R}^{\operatorname{T}}\\
    \vdots &\vdots&\vdots&\vdots\\
    \mathbf S_{L,1}^{\operatorname{T}}&\mathbf S_{L,2}^{\operatorname{T}}&\cdots&\mathbf S_{L,R}^{\operatorname{T}}
    \end{matrix}
    \right]^{\operatorname{T}}
    \in \mathbb C^{RN\times L M_B},
\end{align}}
%Some trouble
we have
\begin{align}\label{eq30}
 \sum_{l=1}^L\sum_{i=1}^K{\mathbf H_{l,k}^{\operatorname{H}}\mathbf W_{l,i}\mathbf W_{l,i}^{\operatorname{H}}\mathbf H_{l,k}}
 =\mathbf D_{k}^{\operatorname{H}}\mathbf W \mathbf D_{k}+\mathbf D_{k}^{\operatorname{H}}\mathbf W \mathbf S^{\operatorname{H}} \Theta^{\operatorname{H}} \mathbf G_k 
 +\mathbf G_k^{\operatorname{H}}\Theta\mathbf S \mathbf W\mathbf D_k
 + \mathbf G_k^{\operatorname{H}}\Theta\mathbf S\mathbf W \mathbf S^{\operatorname{H}}\Theta^{\operatorname{H}}\mathbf G_k,  
\end{align}
where $\mathbf W=\sum_{i=1}^K{\mathbf W_i \mathbf W_i^{\operatorname{H}}}$, and we have 
\begin{align}\label{eq31}
    \sum_{l=1}^L{\mathbf H_{l,k}^{\operatorname{H}}\mathbf W_{l,k}}=\mathbf D_{k}^{\operatorname{H}}\mathbf W_{k}+\mathbf G_{k}^{\operatorname{H}}\Theta\mathbf S \mathbf W_{k}.
\end{align}}
By substituting \eqref{eq30} and \eqref{eq31} into ${f_4}\left( {\mathbf W^{\left(\operatorname{t+1}\right)},\Theta,\mathbf Y^{\left(\operatorname{t+1}\right)} } \right)$, and omitting irrelevant constants with respect to $\Theta$, i.e., $\sum_{k=1}^K{\operatorname{Tr}\left(\mathbf {\bar U}_k \mathbf Y_k^{\operatorname{H}}\mathbf D_{k}^{\operatorname{H}}\mathbf W \mathbf D_{k}\mathbf Y_k\right)}$, $\sum_{k=1}^K{\operatorname{Tr}\left(\mathbf {\bar U}_k \mathbf Y_k^{\operatorname{H}}\mathbf D_{k}^{\operatorname{H}}\mathbf W_k\right)}$, $\sum_{k=1}^K{\operatorname{Tr}\left(\mathbf {\bar U}_k\mathbf W_k^{\operatorname{H}}\mathbf D_{k}\mathbf Y_k\right)}$, and $\sum_{k=1}^K{\operatorname{Tr}\left(\sigma_k^2\mathbf {\bar U}_k\mathbf Y_k^{\operatorname{H}}\mathbf Y_k\right)}$, we have ${f_4}\left( {\mathbf W^{\left(\operatorname{t+1}\right)},\Theta,\mathbf Y^{\left(\operatorname{t+1}\right)} } \right)\triangleq f_6\left(\Theta\right)+\operatorname{Consts}\left(\Theta\right)$, where $f_6\left(\Theta\right)$ can be simplified expressed as follows
\begin{align}\label{eq32}
    f_6\left(\Theta\right)&=\sum_{k=1}^K{\operatorname{Tr}\left(\mathbf {\bar U}_k \mathbf Y_k^{\operatorname{H}} \mathbf G_k^{\operatorname{H}}\Theta\mathbf S\mathbf W_k\right)}+\sum_{k=1}^K{\operatorname{Tr}\left(\mathbf {\bar U}_k \mathbf W_k^{\operatorname{H}}\mathbf S^{\operatorname{H}}\Theta^{\operatorname{H}}\mathbf G_k \mathbf Y_k\right)}\notag\\
    &-\sum_{k=1}^K{\operatorname{Tr}\left(\mathbf {\bar U}_k \mathbf Y_k^{\operatorname{H}}\mathbf G_k^{\operatorname{H}}\Theta\mathbf S\mathbf W \mathbf S^{\operatorname{H}}\Theta^{\operatorname{H}}\mathbf G_k\mathbf Y_k\right)}-\sum_{k=1}^K{\operatorname{Tr}\left(\mathbf {\bar U}_k \mathbf Y_k^{\operatorname{H}}\mathbf G_k^{\operatorname{H}}\Theta\mathbf S \mathbf W\mathbf D_k\mathbf Y_k\right)}\notag\\
    &-\sum_{k=1}^K{\operatorname{Tr}\left(\mathbf {\bar U}_k \mathbf Y_k^{\operatorname{H}}\mathbf D_{k}^{\operatorname{H}}\mathbf W \mathbf S^{\operatorname{H}}\Theta^{\operatorname{H}} \mathbf G_k\mathbf Y_k\right)}.
\end{align}
Consequently, we have the following equivalent problem
\begin{align}\label{eq33}
    \mathop{\max}_\Theta \; f_6\left(\Theta\right),\quad
    \operatorname{s.t.} \;\eqref{eq9}.
\end{align}
To obtain the solution of $\Theta$, we transform Problem \ref{eq33} into an equivalent CMC-QP form. First, by defining $\mathbf Z_k = \mathbf G_k \mathbf Y_k \mathbf {\bar U}_k \mathbf Y_k^{\operatorname{H}} \mathbf G_k^{\operatorname{H}}$, $\mathbf Z =\sum_{k=1}^K{\mathbf Z_k}$, and $\mathbf Q= \mathbf S \mathbf W \mathbf S^{\operatorname{H}}$, we have 
\begin{align}
   \sum_{k=1}^K{\operatorname{Tr}\left(\mathbf {\bar U}_k \mathbf Y_k^{\operatorname{H}}\mathbf G_k^{\operatorname{H}}\Theta\mathbf S\mathbf W \mathbf S^{\operatorname{H}}\Theta^{\operatorname{H}}\mathbf G_k\mathbf Y_k\right)}=\operatorname{Tr}\left(\Theta^{\operatorname{H}}\mathbf Z \Theta \mathbf Q\right).
\end{align}
Then, by defining $\mathbf A_k =\mathbf G_k\mathbf Y_k\mathbf {\bar U}_k\mathbf Y_k^{\operatorname{H}}\mathbf D_k^{\operatorname{H}}\mathbf W\mathbf S^{\operatorname{H}}$, we have 
\begin{align}
    \sum_{k=1}^K{\operatorname{Tr}\left(\mathbf {\bar U}_k \mathbf Y_k^{\operatorname{H}}\mathbf D_{k}^{\operatorname{H}}\mathbf W \mathbf S^{\operatorname{H}}\Theta^{\operatorname{H}} \mathbf G_k\mathbf Y_k\right)}=\operatorname{Tr}\left(\Theta^H\mathbf A\right),\;
    \sum_{k=1}^K{\operatorname{Tr}\left(\mathbf {\bar U}_k \mathbf Y_k^{\operatorname{H}}\mathbf G_k^{\operatorname{H}}\Theta\mathbf S \mathbf W\mathbf D_k\mathbf Y_k\right)}=\operatorname{Tr}\left(\mathbf A^H\Theta\right).
\end{align}
Next, by defining $\mathbf E_k=\mathbf G_k\mathbf Y_k\mathbf {\bar U}_k\mathbf Y_k^{\operatorname{H}}\mathbf S^{\operatorname{H}}$, we have 
\begin{align}
    \sum_{k=1}^K{\operatorname{Tr}\left(\mathbf {\bar U}_k \mathbf W_k^{\operatorname{H}}\mathbf S^{\operatorname{H}}\Theta^{\operatorname{H}}\mathbf G_k \mathbf Y_k\right)}=\operatorname{Tr}\left(\Theta^H\mathbf E \right),\;\sum_{k=1}^K{\operatorname{Tr}\left(\mathbf {\bar U}_k \mathbf Y_k^{\operatorname{H}} \mathbf G_k^{\operatorname{H}}\Theta\mathbf S\mathbf W_k\right)}=\operatorname{Tr}\left(\mathbf E ^H\Theta\right).
\end{align}
As such, we can equivalently transform the objective function $f_6\left(\Theta\right)$ as 
\begin{align}
    f_6\left(\Theta\right)=\operatorname{Tr}\left(\Theta^{\operatorname{H}}\left(\mathbf E - \mathbf A\right) \right)
    +\operatorname{Tr}\left(\left(\mathbf E - \mathbf A\right)^{\operatorname{H}} \Theta\right)
    -\operatorname{Tr}\left(\Theta^{\operatorname{H}}\mathbf Z \Theta \mathbf Q\right).
\end{align}
Additionally, we have a sequence of equations \cite{2017Matrix} as follows
%\begin{lemma}
\begin{align}
    \operatorname{Tr}\left(\Theta^{\operatorname{H}}\mathbf Z \Theta \mathbf Q\right)=\theta^{\operatorname{H}}\mathcal Z\theta, \operatorname{Tr}\left(\Theta^{\operatorname{H}}\Omega \right)= \theta^{\operatorname{H}}\omega,
\end{align}
where $\mathcal Z \triangleq \mathbf Z \odot \mathbf Q^{\operatorname{T}}\in \mathbb C^{\mathcal N\times \mathcal N}$, $\Omega=\mathbf E - \mathbf A\in \mathbb C^{\mathcal N\times \mathcal N}$ with $\mathcal N=RN$, and
\begin{align}\label{eq39}
    \theta=\left(\Theta_{1,1},\cdots, \Theta_{1,N},\Theta_{2,1},\cdots,\Theta_{R,N}\right)^{\operatorname{T}}\in \mathbb C^{\mathcal N\times 1},\notag\\
    \omega=\left[\Omega_{1,1},\cdots,\Omega_{N,N},\Omega_{N+1,N+1},\cdots,\Omega_{\mathcal N,\mathcal N}\right]^{\operatorname{T}}\in \mathbb C^{\mathcal N\times 1}.
\end{align}
%\end{lemma}
%\begin{proof}
%The sequence of equalities given above follows from the following properties of $\operatorname{vec}\left(\cdot\right)$ and $\operatorname{Tr}\left(\cdot\right)$ operators for any matrices $\mathbf O$, $\mathbf P$ and $\mathbf Q$ of appropriate dimensions:
%\begin{align}
%    \operatorname{Tr}\left(\Theta ^{\operatorname{H}}\Lambda\right)=\operatorname{vec}\left(\Theta\right)^{\operatorname{H}}\operatorname{vec}\left(\Lambda\right),\notag\\
%    \operatorname{Tr}\left( \Lambda^{\operatorname{H}}\Theta\right)=\operatorname{vec}\left(\Lambda\right)^{\operatorname{H}}\operatorname{vec}\left(\Theta\right),\notag\\
%    \operatorname{Tr}\left(\Theta^{\operatorname{H}}\mathbf Z \Theta \mathbf Q\right)=\operatorname{vec}\left(\Theta\right)^{\operatorname{H}}\operatorname{vec}\left(\mathbf Z \Theta \mathbf Q\right),\notag\\
%    \operatorname{vec}\left(\mathbf O\mathbf P\mathbf Q\right)=\left(\mathbf Q^{\operatorname{T}}\otimes \mathbf O\right) \operatorname{vec}\left(\mathbf P\right)
%\end{align}
%\end{proof}
 Based on the aforementioned definitions, we can transform Problem \ref{eq33} equivalently as follows 
% \begin{align}\label{eqor}
%     \mathop{\min}_\theta &\quad \theta^{\operatorname{H}}\mathcal Z\theta-\theta^{\operatorname{H}}\omega-\omega^{\operatorname{H}}\theta\\
%     \operatorname{s.t.} &\quad\left|\theta_{i}\right|={\alpha},i =1,2,\cdots, RN.\tag{46a}\label{eq46a}
% \end{align}
\begin{align}\label{eq511}
    \mathop{\max}_\theta &\quad f_7\left(\theta\right)=-\theta^{\operatorname{H}}\mathcal Z\theta+\theta^{\operatorname{H}}\omega+\omega^{\operatorname{H}}\theta\\
    \operatorname{s.t.} &\quad \left|\theta_{i}\right|={\alpha},i =1,2,\cdots, \mathcal N,\tag{46a}\label{eq46a}\notag
\end{align}
 Note that the above CMC-QP problem is non-convex due to the non-convexity of the constant modulus constraints \eqref{eq46a}. As follows, we provide a pair of relaxation methods to solve the non-convex CMC-QP problem.

\subsubsection{SDR-based algorithm}\textcolor{blue}{
By defining $\mathcal {\bar N} ={\mathcal N}+1$, we reformulate the above problem as
\begin{align}\label{eq41}
    \mathop{\max}_\theta &\quad \hat\theta^{\operatorname{H}}\bar{\mathcal { Z}}\hat\theta\\
    \operatorname{s.t.} &\quad\left|\hat\theta_{i}\right|={\alpha}, i =1,2,\cdots, \mathcal N,\tag{47a}
\end{align}
where $\hat\theta=\left[\theta^{\operatorname{T}},\alpha\right]^{\operatorname{T}}\in \mathbb C^{\mathcal {\bar N}\times 1}$, and 
\begin{align}
\bar{\mathcal { Z}}=\left[
\begin{matrix}
-\mathcal { Z}&\omega\\
\omega^{\operatorname{H}}&0
\end{matrix}
\right]\in \mathbb C^{\mathcal {\bar N}\times \mathcal {\bar N}}.
\end{align}
The above problem can be equivalently reformulated as 
\begin{align}\label{eq43}
     \mathop{\max}_{\hat\theta} &\quad \hat\theta^{\operatorname{H}}\left(\lambda_{\mathcal{\bar Z}}^{\max}\mathbf I_{\mathcal {\bar N}}-\bar{\mathcal { Z}}\right)\hat\theta\\
    \operatorname{s.t.} &\quad\left|\hat\theta_{i}\right|={\alpha}, i =1,2,\cdots, \mathcal {\bar N}, \tag{49a}\label{eq43a}
\end{align}
where $\lambda_{\mathcal {\bar Z}}^{\max}$ is the maximal eigenvalue of $\bar{\mathcal { Z}}$. Since $\hat\theta^{\operatorname{H}}\hat\theta={\alpha^2}{\mathcal{{\bar N}}}$ holds, we have $\hat\theta^{\operatorname{H}}\left(\lambda_{\mathcal {\bar Z}}^{\max}\mathbf I_{\mathcal {\bar N}}\right)\hat\theta=\alpha^2\lambda_{\mathcal {\bar Z}}^{\max}\mathcal{{\bar N}}$, which has no impact on updating $\hat \theta$. Therefore, Problem \ref{eq41} and \ref{eq43} are equivalent.
Let $\hat{\mathcal { Z}}=\lambda_{\mathcal {\bar Z}}^{\max}\mathbf I_{\mathcal {\bar N}}-\bar{\mathcal { Z}}$, and it can be readily checked that $\hat{\mathcal Z}$ is a semi-definite matrix. Therefore, the problem can be solved by exploiting the SDR related techniques.
By defining a new variable matrix $\mathcal V=\hat\theta\hat\theta^{\operatorname{H}} \in \mathbb C^{\mathcal {\bar N}\times \mathcal {\bar N}}$, which satisfies $\mathcal V \succeq 0$ and $\operatorname{Rank}\left(\mathcal V\right)=1$. By dropping the non-convex rank-one constraint, we have the following convex semi-definite programming (SDP) problem 
\begin{align}\label{eqsdr}
     \mathop{\max}_{\mathcal V} &\quad \operatorname{Tr}\left(\hat{\mathcal{Z}}\mathcal V\right)\\
    \operatorname{s.t.} &\quad\mathcal V_{i,i}={\alpha^2}, i =1,2,\cdots, \mathcal {\bar N},\tag{50a}\\
    &\quad\mathcal V \succeq 0.\tag{50b}
\end{align}
The above problem can be solved by exploiting the convex solver tools, e.g., CVX \cite{cvx}.
Since the returned solution, $\mathcal V$ fails to be guaranteed rank-one, thus need to adopt the Gaussian randomization approach \cite{1634819} to obtain locally optimal solutions.  
Expressing the eigenvalue decomposition of $\mathcal V$ as $\mathcal V=\hat{\mathbf U}_{\mathcal V}\Sigma_{\mathcal V}\hat{\mathbf U}_{\mathcal V}^{\operatorname{H}}$, we then set $\varpi=\hat{\mathbf U}_{\mathcal V}\Sigma_{\mathcal V}^{\frac{1}{2}}\zeta\in \mathbb C^{\mathcal {\bar N} \times 1}$, where the vector $\zeta \sim \mathcal C \mathcal N \left(\mathbf 0, \mathbf I_{\mathcal {\bar N}}\right)$ is selected for maximizing $\varpi^{\operatorname{H}}\left(\lambda_{\mathcal{\bar Z}}^{\max}\mathbf I_{\mathcal {\bar N}}-\bar{\mathcal { Z}}\right)\varpi$.
Thus, we obtain the locally optimal solution of Problem \ref{eq511} as $\theta=\alpha e^{j\arg\left(\varpi_{\left[1:\mathcal N\right]}/\varpi_{\mathcal {\bar N}}\right)}$, which satisfies the constant modulus constraints.}
% The details of SDR-based approximate solution are presented in Algorithm \ref{asdr}.

% \begin{algorithm}[t]
% \caption{SDR-based Algorithm for Optimizing the Passive Reflecting Beamforming} 
% \label{asdr} 
% \begin{algorithmic}[1]
% \REQUIRE $\hat{\mathcal{Z}}$, $\hat\theta$.
% \STATE \textbf{Defining} $\mathcal V=\hat\theta\hat\theta^{\operatorname{H}} \succeq 0$
% \STATE \textbf{Update}  $\mathcal V$ by solving SDP problem \ref{eqsdr};
% \STATE \textbf{Decompose} $\mathcal V$ as $\mathcal V=\hat{\mathbf U}_{\mathcal V}\Sigma_{\mathcal V}\Psi _{\mathcal V}^{\operatorname{H}}$;
% \STATE \textbf{Select} $\varpi=\hat{\mathbf U}_{\mathcal V}\Sigma_{\mathcal V}^{\frac{1}{2}}\zeta$ for maximizing Problem \ref{eq43};
% \ENSURE locally optimal solution   $\theta^{\operatorname{\left(t+1\right)}}=\alpha e^{j\arg\left(\varpi_{\left[1:\mathcal N\right]}/\varpi_{\mathcal {\bar N}}\right)}$.
% \end{algorithmic} 
% \end{algorithm}
\subsubsection{QCR-based algorithm}\textcolor{blue}{Note that the objective function of Problem \ref{eq511} is convex with regard to $\theta$, and the non-convexity of the problem is caused by the constant modulus constraints \eqref{eq46a}. We consider relaxing the constraints as convex forms, i.e., $\left|\theta_i\right|=\alpha  \to \theta^{\operatorname{H}} \mathbf e_i \mathbf e_i^{\operatorname{H}} \theta \le \alpha, \forall i \in \mathcal N$.
% \begin{align}
%     \left|\theta_i\right|=\alpha  \to \theta^{\operatorname{H}} \mathbf e_i \mathbf e_i^{\operatorname{H}} \theta \le \alpha, \forall i \in \mathcal N.
% \end{align}
Thus, Problem \ref{eq511} is relaxed as a standard convex QCQP problem as follows
\begin{align}\label{eq51aa}
    \mathop{\max}_\theta &\quad f_7\left(\theta\right)\\
    \operatorname{s.t.} &\quad \theta^{\operatorname{H}} \mathbf e_i \mathbf e_i^{\operatorname{H}} \theta \le \alpha, \forall i \in \mathcal N, \tag{51a}\notag
\end{align}
which can be solved by relying on the generic convex solver, e.g. CVX, or employing the Lagrangian sub-gradient dual based algorithm with the locally optimal solutions in nearly closed-forms. Then, we have the following lemma
\begin{lemma}
The relaxed solution generated by the QCR-based algorithm satisfies the Karush-Kuhn-Tucker (KKT) conditions of Problem \ref{eq511}.
\end{lemma}
\begin{proof}
 Based on the fact that the KKT conditions of Problem \ref{eq51aa} constitute exactly that of Problem \ref{eq511}, the lemma can be easily proved. Hence, it is omitted for simplicity.
\end{proof}}

%it is worth pointing out that the loss caused by the relaxation operation is large and hard to estimate.
% Due to the dual gap is guaranteed to be zero, the relaxed problem also can solved by employing Algorithm \ref{a1}, and we have the optimal solution as 
% \begin{align}
%     \theta = \left({\mathcal {\hat Z}+\sum_{i=1}^{\mathcal N}{\upsilon_n \mathbf e_i \mathbf e_i^{\operatorname{H}}}}\right)^{\operatorname{-1}}{\omega},
% \end{align}
% where $\upsilon=\left[\upsilon_1,\upsilon_2,\cdots,\upsilon_{\mathcal N}\right]$ can be determined by ellipsoid method. 
However, it is worth pointing out SDR and QCR suffer from high computational complexities (i.e., $\mathcal O\left(\mathcal {\bar N}+1+\mathcal {\bar N}^2\right)^{3.5}$ for SDR \cite{9279253} and $\mathcal O\left(\mathcal {N}^6\right)$ for QCR \cite{8982186}), especially for the large-scale IRS. %Meanwhile, the loss caused by the relaxation operations is hard to estimate. 
To this end, we propose a computationally efficient ASO algorithm \cite{7946256} to design the passive reflecting beamforming in the next section.

%Now, we prove that the converged solution, i.e., $\hat\theta^{\operatorname{\left(t+1\right)}}$ satisfies the KKT conditions of Problem \ref{eq43}. The corresponding 

\section{Computational Efficient ASO algorithm }
\textcolor{blue}{Note that $f_7\left(\theta\right)$ and the constant modulus constraints \eqref{eq46a} are separable with respect to $\theta_i,\forall i \in \mathcal N$ \cite{7946256}, therefore, we can decompose Problem \ref{eq511} into $\mathcal N$ separate subproblems and solve them one-by-one. Particularly, we have 
\begin{align}
    \theta^{\operatorname{H}}\omega=\sum_{n=1}^{\mathcal N}{\theta_n^\ast\omega_n}=\theta_i^\ast\omega_i +\sum_{n=1,n\ne i}^{\mathcal N}{\theta_n^{\ast}\omega_n}.
\end{align}
Meanwhile, $\theta^{\operatorname{H}} {\mathcal { Z}} \theta$ can be expanded as
\begin{align}\label{eq44}
     \theta^{\operatorname{H}} {\mathcal { Z}} \theta&=\sum_{n=1\atop n\ne i}^{\mathcal N}{ \theta^{\operatorname{H}}\mathbf z _n \theta_n}+{ \theta^{\operatorname{H}}\mathbf z _i \theta_i}=\sum_{n=1\atop n\ne i}^{\mathcal N}{ \theta_i^{\ast} z _{i,n} \theta_n}+{ \theta^{\operatorname{H}}\mathbf z _i \theta_i}+\sum_{m=1\atop m\ne n}^{\mathcal N}{\sum_{p=1\atop p\ne i}^{\mathcal N}{ \theta_m^{\ast} z _{m,p} \theta_p}}\notag\\
    &= \theta_i^{\ast} z _{i,i} \theta_i+\sum_{n=1\atop n\ne i}^{\mathcal N}{\left( \theta_i^{\ast} z _{i,n} \theta_n+ \theta_i z _{n,i}\theta_n^{\ast}\right)}+\sum_{m=1\atop m\ne n}^{\mathcal N}{\sum_{p=1\atop p\ne i}^{\mathcal N}{ \theta_m^{\ast} z _{m,p} \theta_p}},
\end{align}
where $ {\mathcal Z} = \left[\mathbf z_1,\mathbf z_2,\cdots,\mathbf z_{\mathcal N}\right]$ and $\mathbf z_n =\left[z_{1,n}, z_{2,n},\cdots,z_{\mathcal N,n}\right]^{\operatorname{T}}\in \mathbb C^{\mathcal{N}\times 1}$. By using the property $z_{i,n}=z_{n,i}^{\ast}$ and basing the fact that $ {\mathcal { Z}}$ is a positive semi-definite matrix, we have
\begin{align}
     &-\theta^{\operatorname{H}} {\mathcal { Z}} \theta
    +\theta^{\operatorname{H}}\omega+\omega^{\operatorname{H}}\theta\notag\\
    &= 2\operatorname{Re}\left\{\theta_i^{\ast}\omega_i+\sum_{n=1,n\ne i}^{\mathcal N}{\theta_n^{\ast}\omega_n}\right\} -\theta_i^{\ast} z _{i,i} \theta_i-2\operatorname{Re}\left\{\sum_{n=1\atop n\ne i}^{\mathcal N}{ \theta_i^{\ast} z _{i,n} \theta_n}\right\}-\sum_{m=1\atop m\ne n}^{\mathcal N}{\sum_{p=1\atop p\ne i}^{\mathcal N}{ \theta_m^{\ast} z _{m,p} \theta_p}}.
\end{align}
Therefore, $f_7\left(\theta\right)$ can be reformulated as
\begin{align}
    f_8\left(\theta\right)=2\operatorname{Re}\left\{\theta_i^{\ast}\mu_i\right\}+\xi, 
\end{align}
where 
\begin{align}
    \mu_{i}=\omega_i - \sum_{n=1\atop n\ne i}^{\mathcal N}{ z _{i,n} \theta_n},\quad
    \xi=2\operatorname{Re}\left\{\sum_{n=1,n\ne i}^{\mathcal N}{\theta_n^{\ast}\omega_n}\right\}-\sum_{m=1\atop m\ne n}^{\mathcal N}{\sum_{p=1\atop p\ne i}^{\mathcal N}{ \theta_m^{\ast} z _{m,p} \theta_p}}-\theta_i^{\ast} z _{i,i} \theta_i,
\end{align}
% \begin{align}
%     \xi _{i}&= \theta_i^{\ast} z _{i,i} \theta_i,\label{eq46}\\
%     \mu_{i}&=\omega_i - \sum_{n=1\atop n\ne i}^{\mathcal N}{ z _{i,n} \theta_n},\label{eq47}\\
%     \rho_i&=\sum_{m=1\atop m\ne n}^{\mathcal N}{\sum_{p=1\atop p\ne i}^{\mathcal N}{ \theta_m^{\ast} z _{m,p} \theta_p}}\label{eq48}.
% \end{align}
where $\xi$ is the irreverent constant term with regard to $\theta_i$ (e.g., $\theta_i^{\ast} z _{i,i} \theta_i=z _{i,i}\left| \theta_i\right|^2={\alpha^2}{z _{i,i}}$), which do not affect the optimal value.
%y substituting $\left| \theta_{i}\right|={\alpha}$ into \eqref{eq46}, we have $\xi_i=z _{i,i}\left| \theta_i\right|^2={\alpha^2}{z _{i,i}}$. Clearly, $\xi_i$ and $\rho_i$ are irreverent constant terms with respect to $ \theta_i$, which do not affect the optimal value. 
Therefore, we can only investigate $\operatorname{Re}\left\{ \theta_i^{\ast}\mu_i\right\}$ for optimizing $ \theta_i$ and sequentially optimize each element while fixing the remaining $\mathcal N-1$ elements. Problem \ref{eq511} can be equivalently transformed as
\begin{align}\label{eq49}
    \mathop{\max}_{ \theta_i}\quad& \operatorname{Re}\left\{ \theta_i^{\ast}\mu_i\right\}\\
    \operatorname{s.t.}\quad&\left| \theta_i\right|={\alpha},\notag
\end{align}}

\textcolor{blue}{An equivalent expression for Problem \ref{eq49} is given by
\begin{align}\label{eq50}
    \mathop{\max}_{ \phi_i}\quad& \cos\left(- \phi_i+\eta_i\right)\\
    \operatorname{s.t.}\quad& \phi_i\in\left[0,2\pi\right],\tag{58a}
\end{align}
where $\eta_i$ and $- \phi_i$ are the phases of $\mu_i$ and $ \theta_i^{\ast}$, respectively. Consequently, Problem \ref{eq50} has a closed-form optimal solution whose phase is given by
\begin{align}
     \phi_i=\eta_i, \forall i \in \mathcal N.
\end{align}
Consequently, we have 
\begin{align}\label{eq52}
     \theta_i={\alpha}e^{j\eta_i}, \forall i \in \mathcal N.
\end{align}
Based on the above discussions, the procedure of sequentially optimizing $ \theta_1, \theta_2,\cdots, \theta_{\mathcal N}$ and then repeatedly until convergence is attained.} The details for optimizing the locally optimal passive reflecting beamforming are summarized in Algorithm \ref{a2}, and we have the following lemma
\begin{lemma}\label{lemmaconver}
Algorithm \ref{a2} is guaranteed to converge.
\end{lemma}
\begin{proof}
The proof is presented in Appendix \ref{lemma2proof}
\end{proof}

Finally, by defining $\hat N=\left(r-1\right)N$, the locally optimal passive reflecting beamforming of the $r$-th IRS in $t$-th iterations is given as
\begin{align}
    \Theta_r^{\operatorname{\left(t+1\right)}}=\operatorname{diag}\left\{ \theta_{  \hat N+1}, \theta_{ \hat  N+2},\cdots, \theta_{  \hat N +N}\right\},\forall r \in R.
\end{align}

\begin{algorithm}[t]
\caption{ASO Algorithm for Optimizing the Passive Reflecting Beamforming} 
\label{a2} 
\begin{algorithmic}[1]
\REQUIRE $ {\mathcal{Z}}$, threshold $\varepsilon_2$, $u=1$,  $\rho^{\operatorname{u}}=f_7\left(\theta^{\operatorname{u}}\right)$.
\STATE \textbf{Update}  $ \theta_i^{\operatorname{u+1}}$ using \eqref{eq52};
\STATE \textbf{Update}  $\rho^{\operatorname{u+1}}$ by substituting $ \theta^{\operatorname{u+1}}$ into $f_7$;
\STATE \textbf{If} $\left|\rho^{\operatorname{u+1}}-\rho^{\operatorname{u}}\right|\le\varepsilon_2$, output $ \theta^{\left(\operatorname{t+1}\right)}= \theta^{\operatorname{u}}$;\quad \textbf{Otherwise}, set $\operatorname{u}=\operatorname{u+1}$ and go to Step 1.
\end{algorithmic} 
\end{algorithm}

\subsection{Extend to Discrete Phase Shift Case}\label{discretecase}
Since the hardware limitations, the phase shifts of IRS can not be continuous, we consider a more practical IRS model, i.e., the discrete phase shift case, where the phase of each element can only take finite discrete value from the discrete phase set and can be implemented by exploiting the PIN diodes technique. The discrete phase shifts are represented by 
\begin{align}
    \Theta_{r,n}^{\operatorname{\mathcal D}}\in \alpha\left\{1,{e^{j2\pi \frac{1}{\mathcal{M}}}},{e^{j2\pi \frac{2}{\mathcal{M}}}},\cdots, {e^{j2\pi \frac{\mathcal{M}-1}{\mathcal{M}}}}\right\},\forall r \in R, \forall n \in N,
\end{align}
where $\mathcal M$ denotes the size of the discrete phase set.

\textcolor{blue}{Note that Algorithm \ref{a2} can solve the discrete phase shift case after minor changes. Thus, we have the following problem, which is similar to Problem \ref{eq50} 
\begin{align}\label{discrete}
    \mathop{\max}_{\hat\phi_i^{\operatorname{\mathcal D}}}\quad& \cos\left(-\hat\phi_i^{\operatorname{\mathcal D}}+\eta_i^{\operatorname{\mathcal D}}\right)\notag\\
    \operatorname{s.t.}\quad&\hat\phi_i^{\operatorname{\mathcal D}}\in\frac{2\pi}{\mathcal{M}}\left\{0,1,\cdots,\mathcal{M}-1\right\}, {i =1,2,\cdots, \mathcal N}.
\end{align}
Inspection of Problem \ref{discrete} reveals that it can be solved by employing an exhaustive search (ES) strategy. The computational complexity is associated with the number of phase shifts and the size of the discrete phase set, i.e., on the order of $\mathcal O\left(\mathcal{M}{\mathcal{N}}^2\right)$.}

\subsection{Analysis of the Overall Algorithm}
\begin{algorithm}[t]
\caption{ASO-based Joint Optimization Algorithm for Solving Problem \ref{eq7}} 
\label{a3} 
\begin{algorithmic}[1]
\REQUIRE    
    $\mathbf W_{l,k} ^{\left ( \operatorname{1} \right )},\forall \left\{l,k\right\}$,
    $\Theta ^{\left ( \operatorname{1} \right )}$;
    threshold $\varepsilon_3 $.
\ENSURE $\mathbf W_{l,k}^{\operatorname{opt}} \triangleq \mathbf W_{l,k}  ^{\left ( \operatorname{t} \right )},\forall {l,k}$; $\Theta_r^{\operatorname{opt}} \triangleq \Theta_r  ^{\left (  \operatorname{t} \right )},\forall r \in R$.\\
\STATE\textbf{Calculate} $\Gamma_{k} ^{\left ( 1 \right )}, \forall k$,\; $\mathcal{R}^{\left( 1 \right )}\left(\mathbf W^{\left ( 1 \right )},\Theta^{\left ( 1 \right )}\right)$;\\
\FOR{$t=1 \;\text{to}\; 2,3,\cdots,t_{\max}$ }
\STATE \textbf{Update} $\mathbf U_{k} ^{\left ( \operatorname{t+1} \right )}, \forall k$, by using \eqref{eq15};
\STATE \textbf{Update} $\mathbf Y_{k} ^{\left ( \operatorname{t+1} \right )}, \forall k$, by using \eqref{eq18};
\STATE \textbf{Update} $\mathbf W_{l,k} ^{\left ( \operatorname{t+1} \right )}, \forall \left\{l,k\right\}$, by employing Algorithm \ref{a1};
\STATE \textbf{Update} $\Theta_r ^{\left ( \operatorname{t+1} \right )},\forall r$, by employing Algorithm \ref{a2};
 \IF {${{\left| {{\mathcal{R}^{\left( \operatorname{t+1} \right)}}\left( {\mathbf W,\Theta } \right) - {\mathcal{R}^{\left( \operatorname{t+1}\right)}}\left( {\mathbf W,\Theta } \right)} \right|}}/{{{\mathcal{R}^{\left( \operatorname{t+1} \right)}}\left( {\mathbf W,\Theta } \right)}} < \varepsilon_3 $} 
\item Break;
 \ENDIF
 \ENDFOR
\end{algorithmic} 
\end{algorithm}
The details of the proposed ASO-based joint optimization algorithm are provided in Algorithm \ref{a3}. 
As follows, we prove the convergence of the ASO-based algorithm. Particularly, according to the properties of LDT and QT methods, we have 
\begin{align}
    \mathcal{R}\left( {{\mathbf W},{\Theta }} \right) = {f_1}\left( {{\mathbf W},{\Theta},{\mathbf U}} \right)={f_3}\left( {{\mathbf W},{\Theta },{\mathbf U}, \mathbf Y} \right).
\end{align}
In $t${-th} iteration, the update rule is designed as 
\begin{align}
   \cdots \mathbf U^{\left(\operatorname{t+1}\right)} \to \mathbf Y^{\left(\operatorname{t+1}\right)} \to \mathbf W^{\left(\operatorname{t+1}\right)} \to \Theta^{\left(\operatorname{t+1}\right)}  \cdots, 
\end{align}
and we have the following inequalities
\begin{align}
    &{f_3}\left( { {\mathbf W^{\left(\operatorname{t}\right)}},{\Theta^{\left(\operatorname{t}\right)} },{\mathbf U}^{\left(\operatorname{t}\right)}, \mathbf Y^{\left(\operatorname{t}\right)}} \right) \notag\\
    \mathop  \le^{\left( a \right)} & {f_3}\left( { {\mathbf W^{\left(\operatorname{t}\right)}},{\Theta^{\left(\operatorname{t}\right)} },{\mathbf U}^{\left(\operatorname{t+1}\right)}, \mathbf Y^{\left(\operatorname{t}\right)}} \right)\notag\\
    \triangleq&{f_4}\left( {{\mathbf W^{\left(\operatorname{t}\right)}},{\Theta^{\left(\operatorname{t}\right)} }, \mathbf Y^{\left(\operatorname{t}\right)}} \right) + \operatorname{Const}\left(\mathbf {\bar U}\right)\notag\\
    \mathop  \le^{\left( b \right)}&{f_4}\left( { {\mathbf W^{\left(\operatorname{t}\right)}},{\Theta^{\left(\operatorname{t}\right)} }, \mathbf Y^{\left(\operatorname{t+1}\right)}} \right) + \operatorname{Const}\left(\mathbf {\bar U}\right)\notag\\
    \triangleq&-{f_5}\left({\mathbf W^{\left(\operatorname{t}\right)}} \right) + \operatorname{Const}\left(\mathbf {Y},\mathbf {\bar U} \right)\notag\\
    \mathop  \le^{\left( c \right)}& -{f_5}\left({\mathbf W^{\left(\operatorname{t+1}\right)}} \right) + \operatorname{Const}\left( \mathbf {Y},\mathbf {\bar U} \right),
\end{align}
where the inequalities, i.e, (a), (b) and (c) hold due to the fact that $\mathbf U$, $\mathbf Y$ and $\mathbf W$ are optimally determined by employing \eqref{eq15}, \eqref{eq18}, and \eqref{eq24}, respectively. Meanwhile, we have 
\begin{align}
    {f_3}\left( {{\mathbf W^{\left(t+1\right)}},{\Theta^{\left(t\right)} },{\mathbf U}^{\left(t+1\right)}, \mathbf Y^{\left(t+1\right)}} \right) 
    \triangleq {f_6}\left({ \Theta^{\left(t\right)}} \right) + \operatorname{Const}\left( \mathbf W,\mathbf {\bar U},\mathbf {Y} \right),
\end{align}
and after some matrix manipulations, we have
\begin{align}
    {f_6}\left({ \Theta^{\left(t\right)}}\right)\triangleq f_7\left(\theta^{\left(t\right)}\right).%\triangleq-f_8\left(\hat\theta^{\left(t\right)}\right)\triangleq f_9\left(\hat\theta^{\left(t\right)}\right)+ c,
\end{align}
%where $c$ denotes the constant term with respect to $\hat\theta$, i.e., $\hat\theta^{\operatorname{H}}\left(\bar\lambda_{\bar{\mathcal { Z}}}\mathbf I_{\mathcal N}\right)\hat\theta$. 
According to the property of Algorithm \ref{a2}, we have 
\begin{align}
    f_7\left(\theta^{\left(t\right)}\right) \le f_7\left(\theta^{\left(t+1\right)}\right).
\end{align}
Above inequalities verify that $\mathcal R\left(\mathbf W, \Theta\right)$ is monotonically non-decreasing after each updating step. Overall, it can be summarized that 
\begin{align}
    \mathcal{R}\left( {{\mathbf W^{\left(t+1\right)}},{\Theta^{\left(t+1\right)} }} \right)\ge \mathcal{R}\left( {{\mathbf W^{\left(t\right)}},{\Theta^{\left(t\right)} }} \right)
    \ge \cdots \ge \mathcal{R}\left( {{\mathbf W^{\left(1\right)}},{\Theta^{\left(1\right)} }} \right).
\end{align}
Besides, $\mathcal{R}\left( {{\mathbf W},{\Theta}} \right)$ is upper-bounded by a finite value due to the limited transmit powers of BSs and the finite number of phase shifts of IRSs. As the number of iterations increases, we finally have $ \mathcal{R}\left( {{\mathbf W^{\operatorname{opt}}},{\Theta^{\operatorname{opt}} }} \right)\triangleq\mathcal{R}\left( {{\mathbf W^{\left( \operatorname{t}_{\max}\right)}},{\Theta^{\left(\operatorname{t}_{\max}\right)} }} \right)$, where $\operatorname{t}_{\max}$ is the maximal number of iterations when Algorithm \ref{a3} converges. Therefore, the strict convergence of Algorithm \ref{a3} can be guaranteed. 

Meanwhile, the complexity of Algorithm \ref{a3} is summarized as follows. The complexities of updating $\mathbf U_k$ and $\mathbf Y_k$ are $\mathcal O \left(M_u^3\right)$, $\forall k$, respectively. The complexity of calculating $\mathbf W_{l,k}$ is $\mathcal O\left(M_b^3\right),\forall {l,k}$, and the complexity of updating $\theta$ is $\mathcal O\left(\mathcal{N}^2\right)$. Therefore, based on the aforementioned discussions, the total complexity of Algorithm \ref{a3} is $\mathcal O\left(\mathcal I_{0}\left(2KM_u^3+\mathcal I_{1}LKM_b^3+\mathcal I_{\operatorname{ASO}}\mathcal{N}^2\right)\right)$, where $\mathcal I_{1}$, $\mathcal I_{\operatorname{ASO}}$ and $\mathcal I_{0}$ are the number of iterations when Algorithm \ref{a1}, Algorithm \ref{a2} and Algorithm \ref{a3} converge, respectively. \textcolor{blue}{For comparison purposes, in Table \ref{com}, we summarize the complexities of SDR, QCR, MM, and ASO algorithms. It is observed that ASO-based algorithm has the lowest computational complexity. Meanwhile, in Section \ref{5.0}, we compare the performance achieved by ASO with regard to the above benchmark schemes.}
\begin{table}
\centering
\caption{Comparison of Complexity}
\label{com}
\begin{tabular}{|c|c|c|c|c|}
\hline 
 {Algorithms}&SDR-based&QCR-based&MM-based&ASO-based\\\hline
 {Computational complexity}&$\mathcal O\left(RN+2+\left({ RN+1}\right)^2\right)^{3.5}$&$\mathcal O\left({R}^6{N}^6\right)$&$\mathcal O\left({R}^3{N}^3\right)+\mathcal O\left(\mathcal I_{\operatorname{MM}}{R}^2{N}^2\right)$&$\mathcal O\left(\mathcal I_{\operatorname{ASO}}{R}^2{N}^2\right)$\\\hline
\end{tabular}
\end{table}

\section{Numerical Simulation and Discussion}{\label{5.0}}
As follows, simulation results are provided to evaluate the performance of the proposed ASO-based joint optimization algorithm. 
We consider a three-IRSs assisted cell-free wireless communication system, where six-BSs are transmitting signals to four-UEs cooperatively. 
As shown in Fig. \ref{f2}, we assume a 3-D scenario, where the height of the BSs, the IRSs, and the UEs are 3m, 6m, and 1.5m, respectively. The four UEs are uniformly and randomly distributed in a circle centered at $(\chi,100)$ with a radius of 10m.
\textcolor{blue}{Meanwhile, we assume the phase shifts along the horizontal and vertical are $N_h=10$ and $N_v=N/N_h$, respectively.} 
Unless otherwise stated, the simulation parameters are summarized in Table \ref{simu}.
\begin{table}
\centering
\caption{Simulation Parameters}
\label{simu}
\begin{tabular}{c|c}
\hline 
 \textbf{Parameters}&\textbf{Values}  \\ \hline
 Number of antennas at each BS&$M_b=4$  \\ \hline
 Number of antennas at each UE&$M_u=2$  \\ \hline
 Number of phase shifts at each IRS&$N=60$  \\ \hline
 Reflecting efficiency of IRSs& $\alpha=1$ \\\hline
 Path loss at the reference distance&$C_0=-30$ dB  \\ \hline
 Maximum transmit power&$P_{\max,l}=0.1$ W, $\forall l$  \\ \hline
 Path loss exponent of direct channel&$p_{lk}=3.75,\forall {l,k}$ \\ \hline
 Path loss exponent of IRS-related channel&$p_{lr}=p_{rk}=2.2,\forall {l,k,r}$ \\ \hline
 Rician factor of IRS-related channel&$\beta_{G}=\beta_{S}=3$ dBw \\ \hline
 Noise power&$\sigma^2=-80$ dBm \\ \hline
\end{tabular}
\end{table}
\begin{figure}
    \centering
 \includegraphics[width=0.5\linewidth]{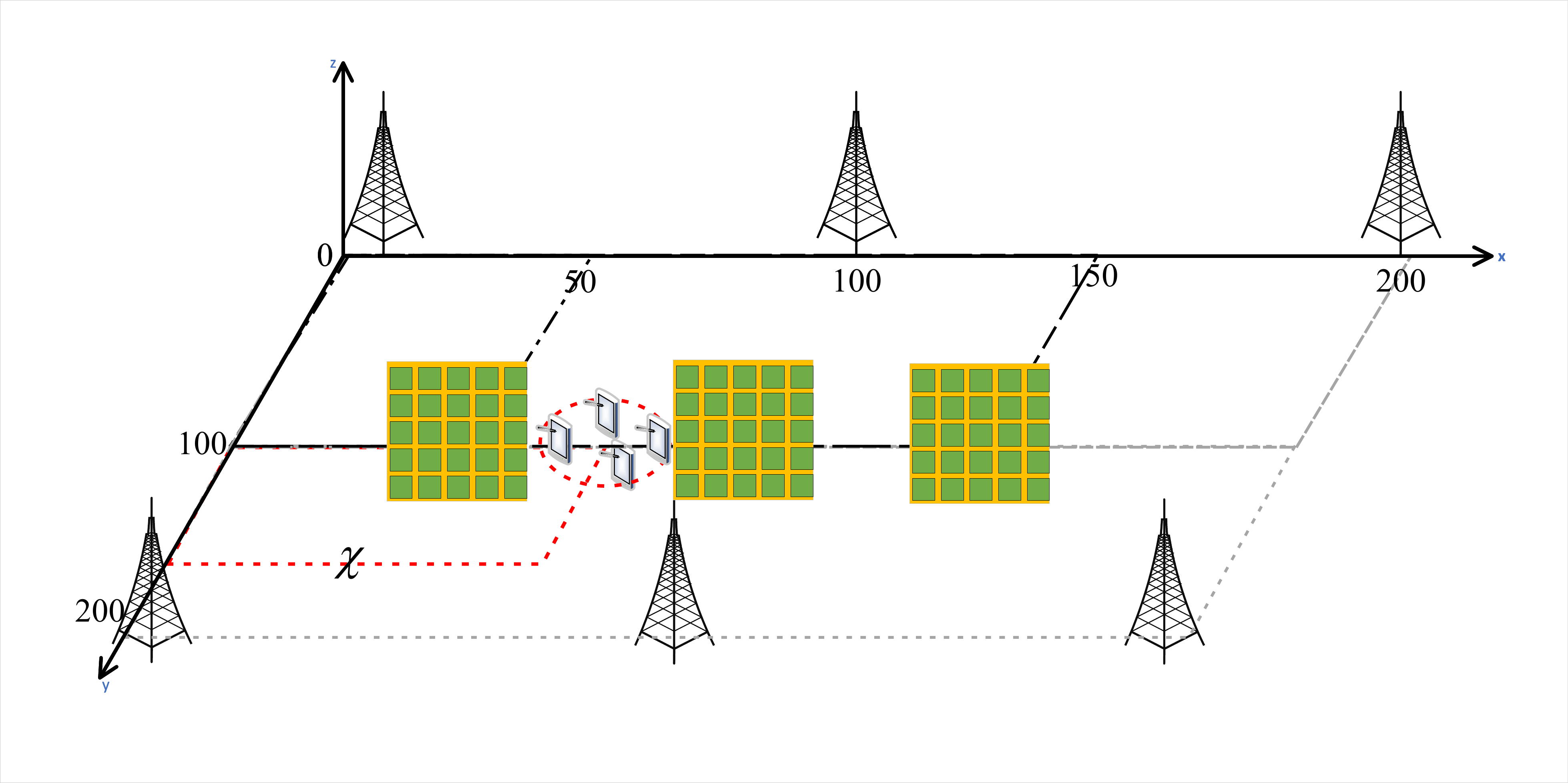}
    \caption{IRSs-assisted cell-free MIMO communication scenario.}
    \label{f2}
\end{figure}

\subsection{Convergence Behavior}
In this subsection, we investigate the convergence of the ASO-based algorithms under the continuous phase shift case and the discrete phase shift case of $\mathcal M=2$ and $\mathcal M=4$, i.e., ``ASO", ``ASO, $\mathcal M=2$", and ``ASO, $\mathcal M=4$".
As shown in Fig. \ref{fconvergence}, we study the sum-rate achieved by the three schemes against the number of iterations under the different number of phase shifts at each IRS, i.e., $N=60$ and $N=120$. 
The curves are consistent with our expectation, where the three schemes converge to stationary points after a few iterations. 
It is observed that the convergence speed of the ``ASO" scheme under $N=120$ is slower than that under $N=60$.
Besides, the curves of the discrete phase shift case also indicate that the convergence speed is sensitive to the size of the discrete phase set. 
% demonstrates the fast convergence of Algorithm \ref{a3}.
\begin{figure}
    \centering
    {\includegraphics[width=0.5\linewidth]{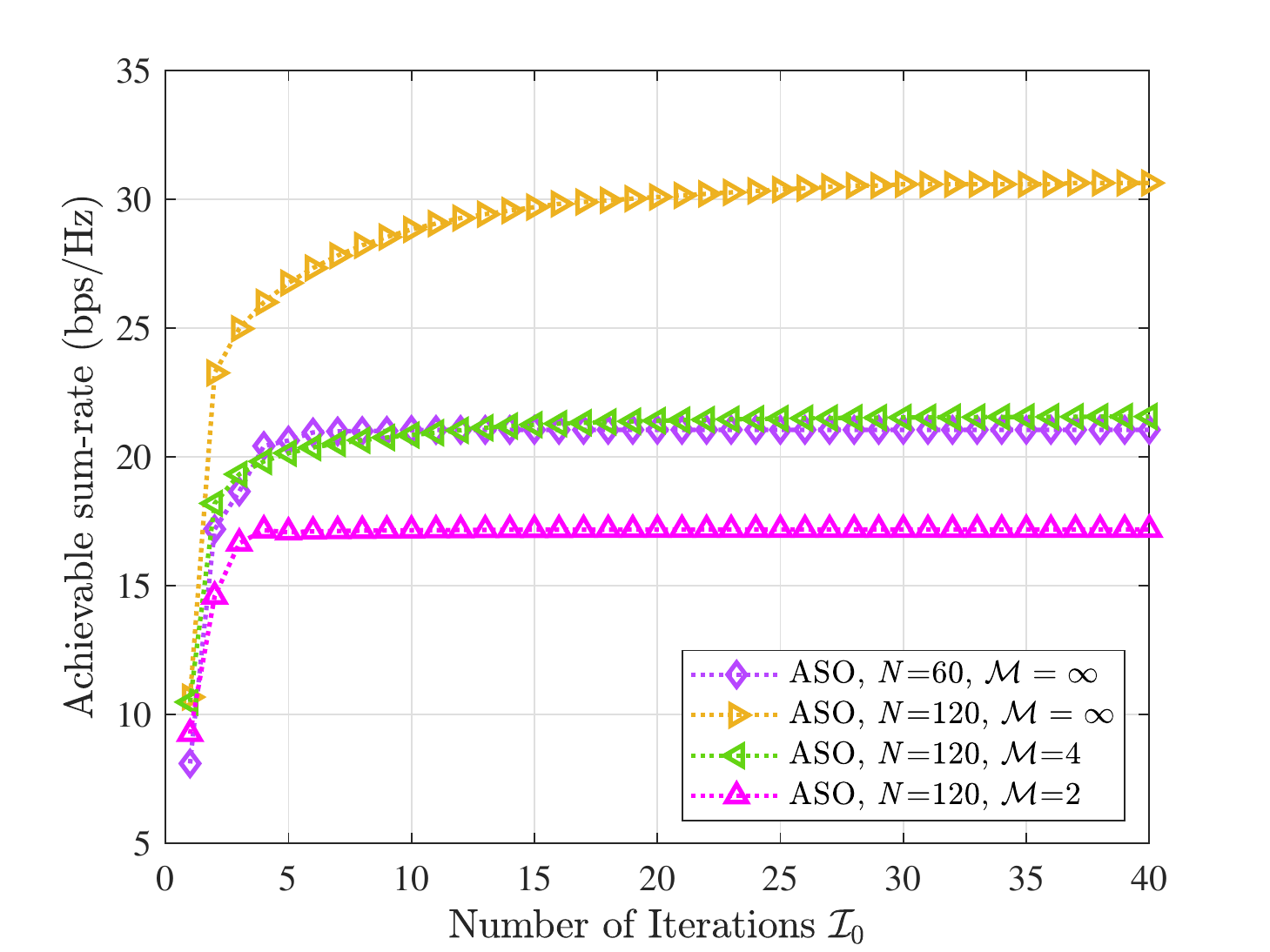}}
    \caption{Achievable sum-rate against the number of iterations.}
 \label{fconvergence} 
\end{figure}

\subsection{Performance Comparison and Discussions}
To show the performance gain achieved by the proposed ASO-based algorithm, we consider the following schemes for comparison
\begin{itemize}
    \item \textbf{SDR/QCR/MM}: \textcolor{blue}{Replacing ASO algorithm in Algorithm \ref{a3} by SDR, QCR and MM techniques, respectively,  for optimizing the continuous phase shift;}
    \item \textbf{random phase}: Only optimizing the active transmit beamforming matrices of BSs with random passive reflecting beamforming matrices of IRSs, where each phase shift is following uniform distribution over $\left [0, 2\pi \right )$;
    \item \textbf{without IRS}: The conventional cell-free scenario without deploying any IRSs.
\end{itemize}
Meanwhile, all the simulation results are obtained by averaging 200 channel realizations. In the following, we investigate the impacts of the critical simulation parameters on the performance gains achieved by the schemes as mentioned above.

\subsubsection{Impact of the CSI estimation error ratio}\label{imcsi}The CSI estimation is challenging to realize in IRS-assisted system, therefore, here we investigate the impact of imperfect CSI on the achievable sum-rate. \textcolor{blue}{We assume a bounded CSI error model \cite{9180053}, i.e.,
\begin{align}
    \mathbf D_{l,k}= \mathbf  {\hat D}_{l,k} + \Delta\mathbf D_{l,k}, \forall {l,k};\quad\mathbf G_{r,k}= \mathbf  {\hat G}_{r,k} + \Delta\mathbf G_{r,k}, \forall {r,k};\quad
    \mathbf S_{l,r}= \mathbf  {\hat S}_{l,r} + \Delta\mathbf S_{l,r}, \forall {l,r},\notag
\end{align}
where $\mathbf  {\hat D}_{l,k}$ ($\mathbf  {\hat G}_{r,k}$ and $\mathbf  {\hat S}_{l,r}$) are estimated CSI and $\Delta\mathbf D_{l,k}$ ($\Delta\mathbf G_{r,k}$ and $\Delta\mathbf S_{l,r}$) are the estimation errors. The norm of the channel uncertainty can be bounded as \cite{6169188}
\begin{align}
    \left\|\Delta\mathbf D_{l,k}\right\|_{\operatorname{F}}\le \rho \left\|\mathbf {\hat D}_{l,k}\right\|_{\operatorname{F}}, \forall {l,k};\quad
    \left\|\Delta\mathbf G_{r,k}\right\|_{\operatorname{F}}\le \rho \left\|\mathbf {\hat G}_{r,k}\right\|_{\operatorname{F}}, \forall {r,k};\quad
    \left\|\Delta\mathbf S_{l,r}\right\|_{\operatorname{F}}\le \rho \left\|\mathbf {\hat S}_{l,r}\right\|_{\operatorname{F}}, \forall {l,r}.\notag
\end{align}}
In Fig. \ref{fCSIerror}, we investigate the sum-rate achieved by all schemes against the error ratio. It can be observed that with increasing the estimation error $\rho$, the performance gaps compare with perfect CSI without error (i.e., $\rho=0$) become larger. Meanwhile, the ``ASO" scheme is robust to the CSI estimation error ratio due to the performance suffers a fewer performance losses compared with the perfect CSI scenario. It is observed that when $\rho=0.1$, the proposed ASO-based algorithm achieves a similar performance to the ``MM" scheme. 
%Particularly, the performance losses are $4.8\%$ when $\rho=0.1$, and $14.4\%$ when $\rho=0.3$.  

It is worth pointing out that the proposed ASO-based algorithm is not the robust IRSs-assisted MIMO cell-free transmission design scheme with the CSI uncertainty, which needs to be investigated and will be left as future work. 

\textcolor{blue}{To evaluate the robustness of the proposed ASO-based algorithm, we add the benchmark scheme ``ASO, $\rho=0.1$", to denote the scheme which exploits the proposed ``ASO" algorithm to solve the problem with uncertainty CSI ($\rho=0.1$).}
\begin{figure}
    \centering
    {\includegraphics[width=0.5\linewidth]{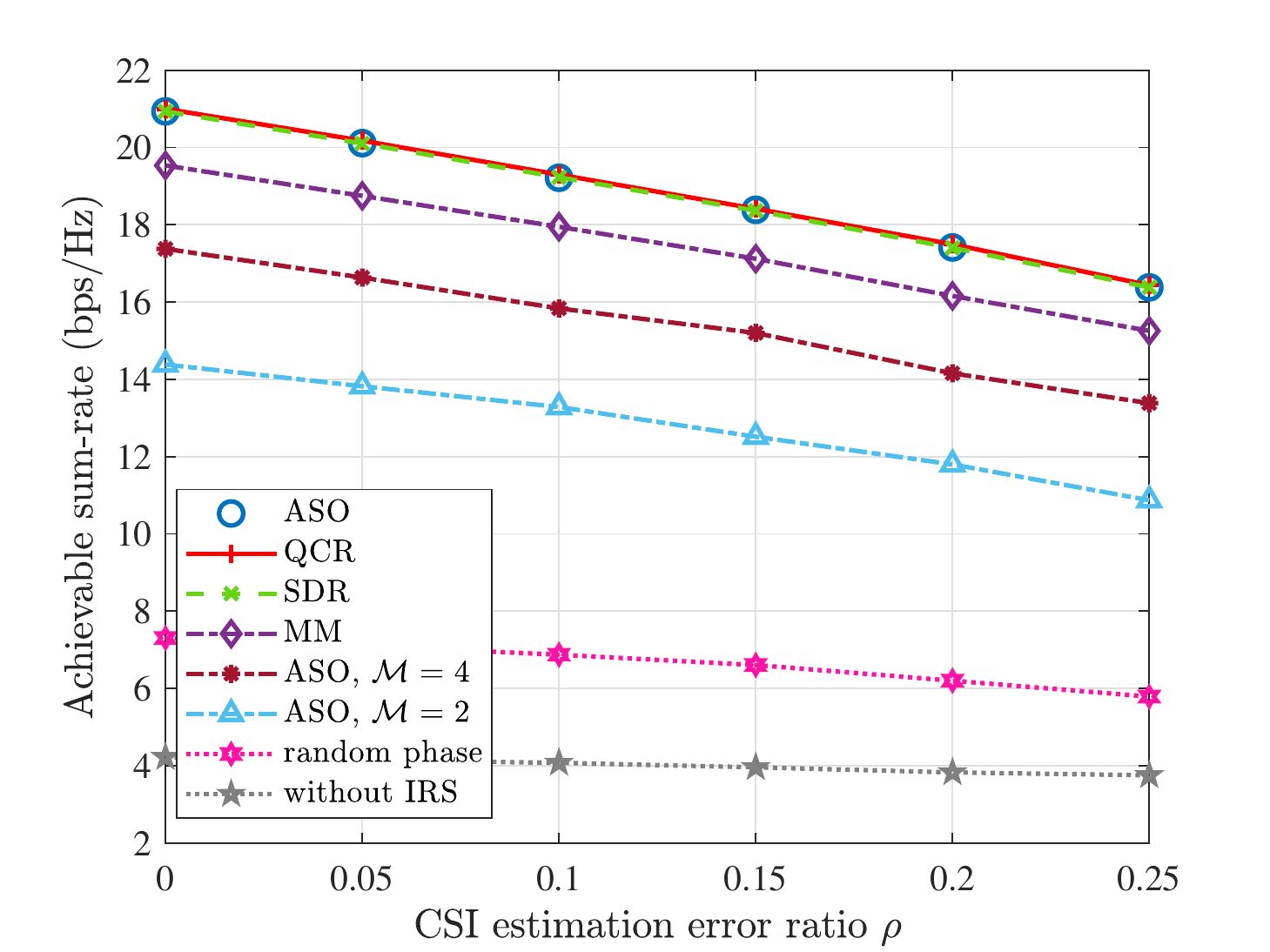}}
    \caption{Achievable sum-rate against the CSI estimation error ratio.}
 \label{fCSIerror} 
\end{figure}

\subsubsection{Impact of the number of phase shifts of each IRS}
We present the sum-rates achieved by all schemes against the number of the phase shifts of each IRS in Fig. \ref{fnirs}.
It is observed that IRSs can considerably improve the performance compared with the conventional cell-free system (i.e., ``without IRS" scheme), wherewith the increasing number of phase shifts, the performance achieved by the IRS-related schemes (i.e., ``ASO", ``ASO, $\rho=0.1$", ``ASO, $\mathcal M=2/4$", ``SDR", ``QCR", ``MM" and ``random phase") increase. 
Meanwhile, the ``random phase" scheme achieves a limited performance gain than the ``without IRS" scheme, demonstrating that the passive reflecting beamforming needs to be carefully optimized to improve the performance significantly.
Besides, the size of the discrete phase set impacts the performance, and with increasing $N$, the performance gaps between ``ASO, $\mathcal M=4$" and ``ASO, $\mathcal M=2$" become larger. The performance loss compared with the continuous phase shift case can be compensated by adopting high-resolution discrete phase shifts.
Most important, the proposed ``ASO" scheme achieves a better or nearly the same performance with regard to ``MM", ``SDR", and ``QCR" schemes, but with the lowest computational complexity. 
\textcolor{blue}{Furthermore, the curve of the ``ASO, $\rho=0.1$" scheme, which achieves a similar performance to that of the ``MM" scheme, which shows our proposed ``ASO" algorithm has the strong robustness to the CSI estimation errors.}

\begin{figure}
    \centering
    {\includegraphics[width=0.5\linewidth]{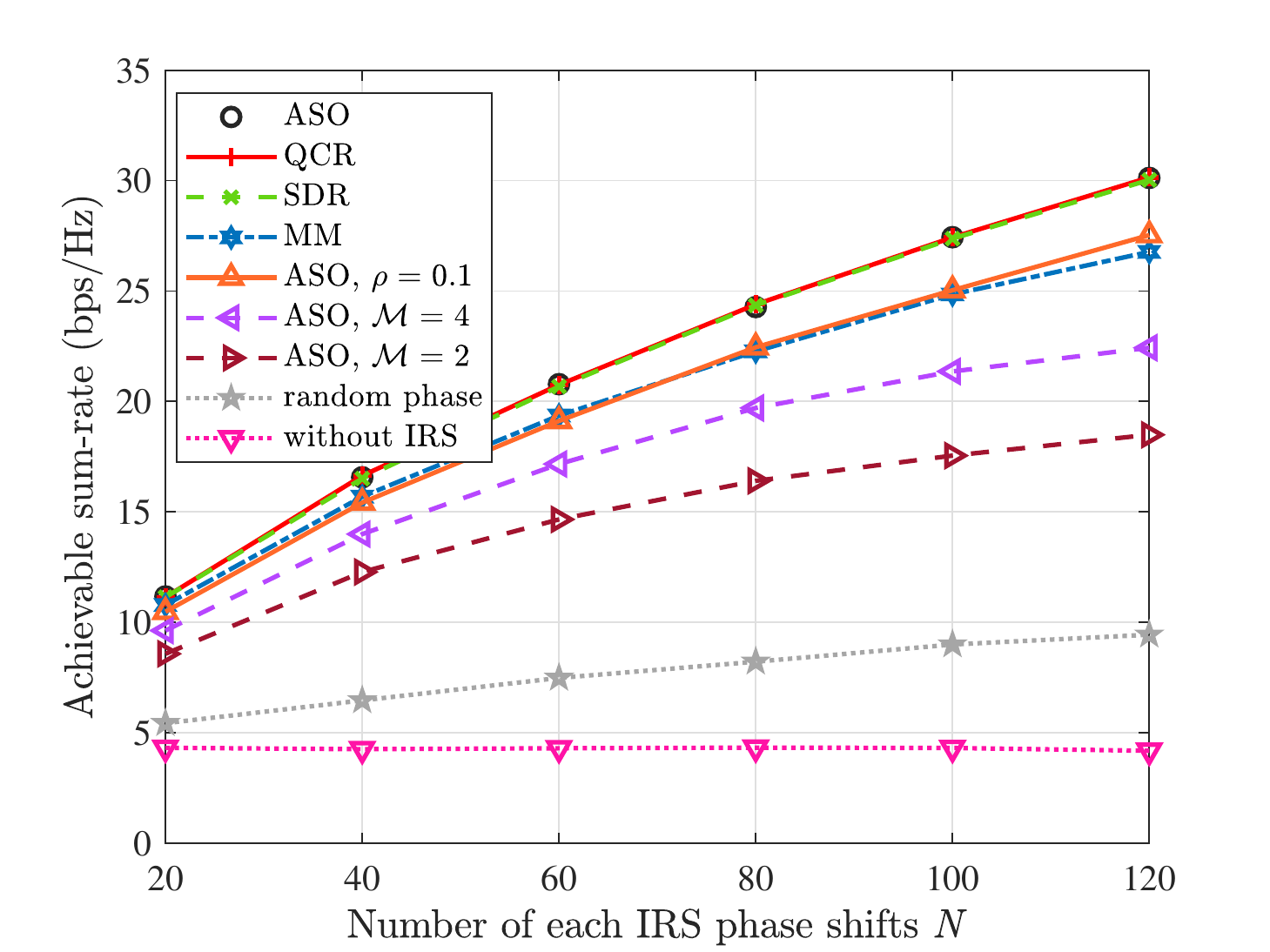}}
    %\scalebox{0.6}{\includegraphics{Nris.pdf}}
    \caption{Achievable sum-rate against the number of phase shifts at each IRS.}
 \label{fnirs} 
\end{figure}

\subsubsection{Impact of the UEs location}
We investigate the impact of the UEs location in Fig. \ref{firsloc}, and the curves are approximately symmetric with respect to the line of $\chi=100$m, which are consistent with the expectations. 
It is observed that as the UEs are deployed closer to each IRS, i.e., $\chi=50$m, $\chi=100$m, and $\chi=150$m, the IRS-related schemes achieve the best performances, which is due to the smaller reflection channel fading. This indicates that the system performance can indeed be improved significantly with the deployment of IRSs, especially when the UEs are closed to the IRS and when the passive reflecting beamforming of IRSs are carefully optimized.

\begin{figure}
    \centering
    %\scalebox{0.5}{ \includegraphics{ce.pdf}}
    {\includegraphics[width=0.5\linewidth]{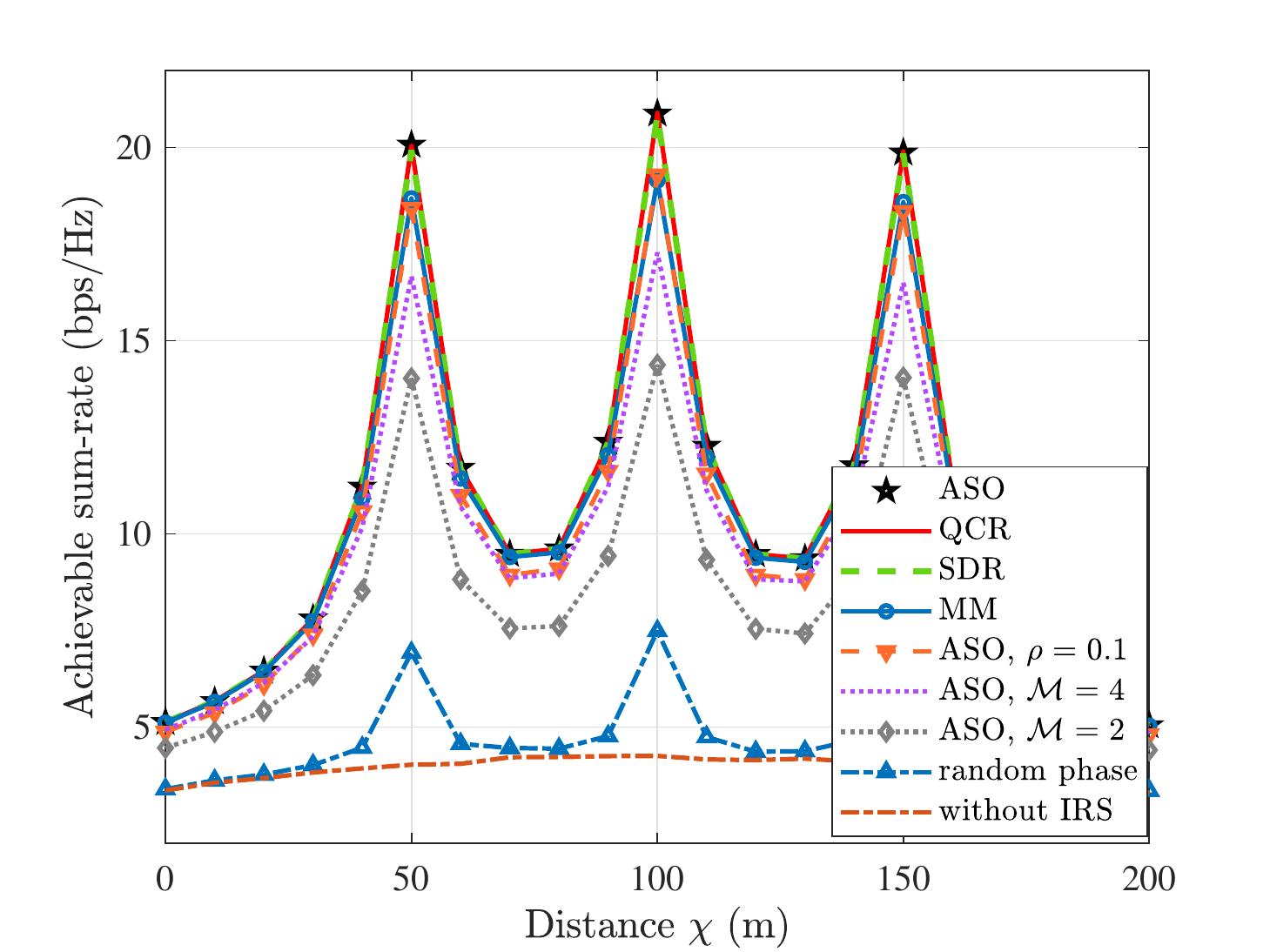}}
    \caption{Achievable sum-rate against the location of the IRS.}
    \label{firsloc}
\end{figure}

\subsubsection{Impact of the path loss exponent}
We investigate the impact of the path loss exponent of the IRS-related channels while fixing that of direct channels. 
As shown in Fig. \ref{fpathloss}, the sum-rate achieved by all IRS-related schemes decrease significantly with the increasing of the path loss exponent, and finally (i.e., $\rho \ge 3.4$), the curves are approximately coinciding with the ``without IRS" scheme. 
This is mainly due to that when the path loss exponents of IRS-related channels are large, the array gains introduced by IRS are negligible. 
\textcolor{blue}{To this end, the location of the IRS should be appropriately chosen for ensuring a free space IRS-related channels can be established.}
\begin{figure}
    \centering
    {\includegraphics[width=0.5\linewidth]{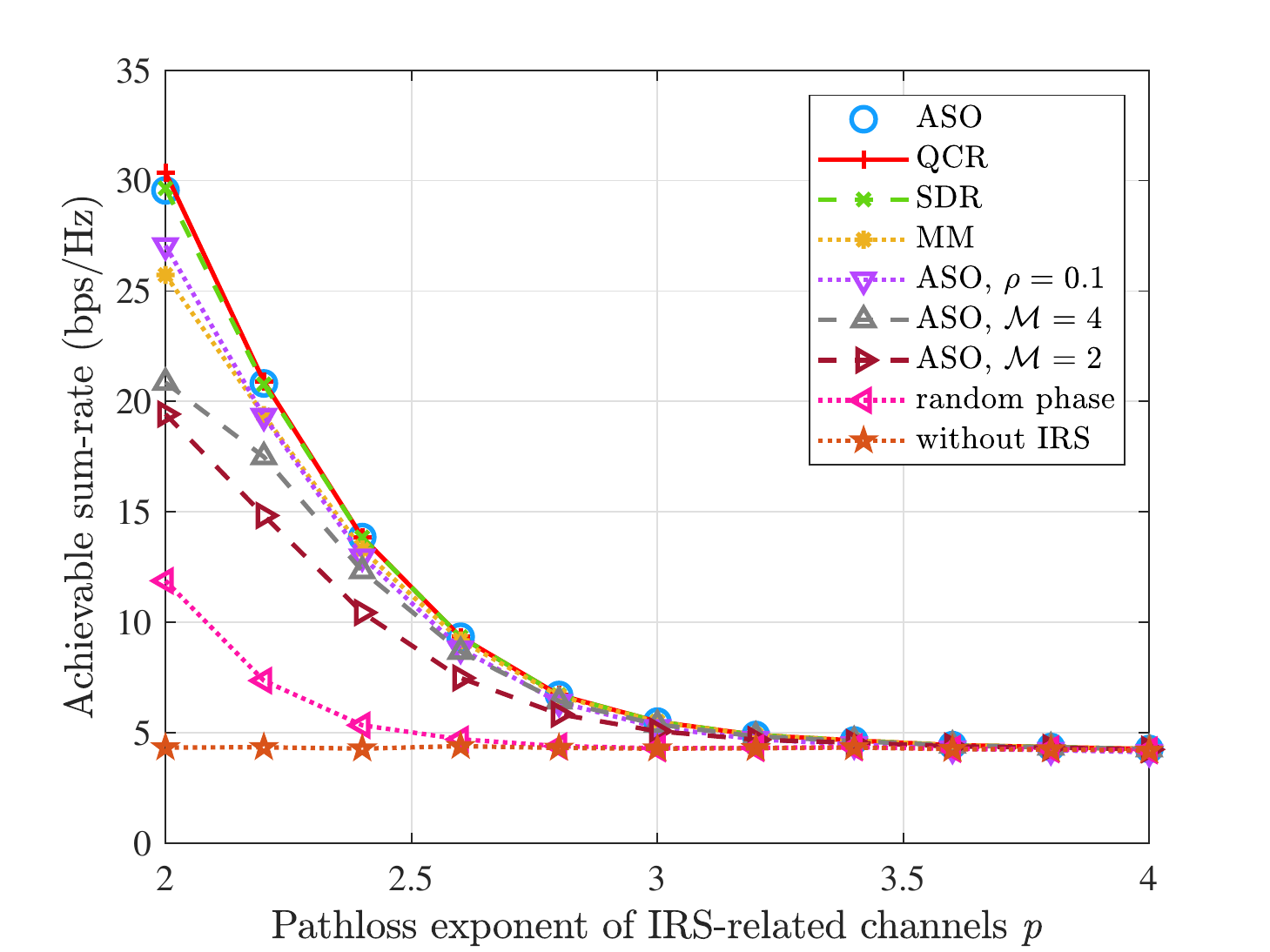}}
    \caption{Achievable sum-rate against the pathloss exponent of the IRS-related channels.}
 \label{fpathloss} 
\end{figure}

\subsubsection{Impact of the reflecting efficiency}
\textcolor{blue}{We present the sum-rates against the reflecting efficiency of IRSs in Fig. \ref{efficiency}. 
It can be observed that the reflecting efficiency of IRS has a substantial impact on the performance, where as expected, with the increasing $\alpha$, the sum-rate achieved by IRS-related schemes increased significantly. 
It can be attributed to that a larger $\alpha$ means the fewer power loss caused by signal absorption at IRSs.}
Meanwhile, some work assumed a more perfect IRS model \cite{9112252,9154244,9459505,8982186}, i.e., the amplitude and phase of each element can be controlled independently, which is expressed as $ \Theta_r=\operatorname{diag}\left(\vartheta _{r,1}e^{j{\phi_{r,1}}},\vartheta_{r,2}e^{j{\phi_{r,2}}},\cdots,\vartheta_{r,N}e^{j{\phi_{r,N}}}\right),\forall r \in R$.
However, the perfect IRS model does not only leads to a better performance gain but also causes a higher implementation cost (computational complexity) with respect to the IRS model assumed in this paper. The trading-off between the performance and the implementation cost (computational complexity) is an interesting problem, which is left for future work.
\begin{figure}
    \centering
    {\includegraphics[width=0.5\linewidth]{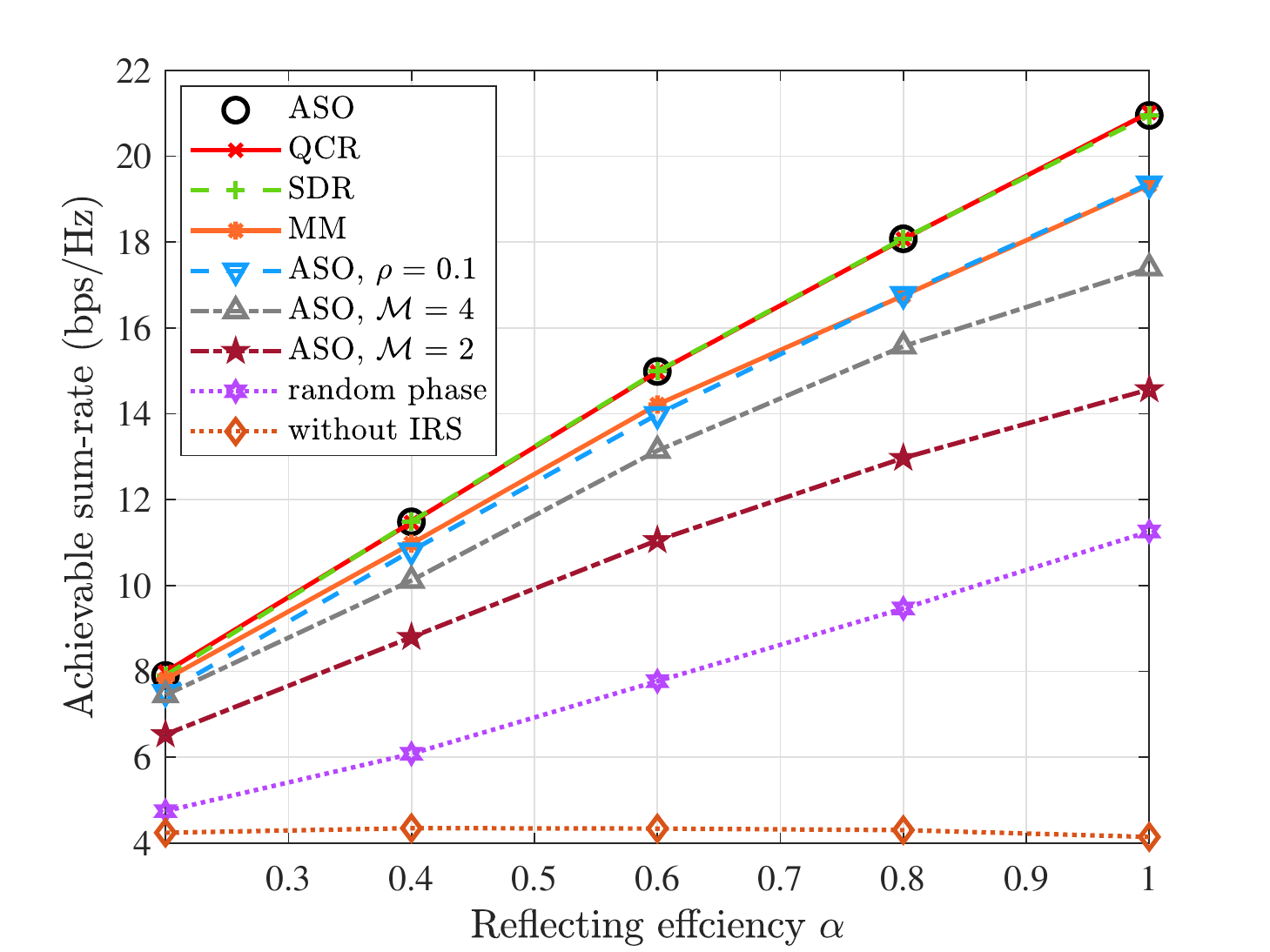}}
    \caption{Achievable sum-rate against the reflecting efficiency.}
 \label{efficiency} 
\end{figure}

\section{Conclusion}{\label{5.0}}
In this paper, we investigated the achievable sum-rate maximization problem in the multiple distributed IRSs assisted cell-free MIMO cooperative transmission system. 
 We proposed an efficient framework to jointly design the BSs, the IRSs, and the UEs. 
 Particularly, we first transformed the non-convex problem to an equivalent form, and then decomposed the reformulated problem into two subproblems, and solved the two subproblems alternating iterative. The proposed algorithms were guaranteed to converge to locally optimal solutions. 
 We also extended the IRSs to the discrete phase shift case and provided an exhaustive search method to solve it. 
Simulation results demonstrated the proposed algorithms achieved considerable performance improvements than the benchmark schemes.
Meanwhile, the robust beamforming design with imperfect CSI and the trading-off between performance with the implementing cost and the computational complexity were interesting problems, which were left as future works.

\appendices
\section{Proof of Proposition \ref{p1}}{\label{ap1}}
By introducing the auxiliary matrix $\mathbf U_k \in \mathbb{C}^{M_u \times M_u},\forall k$, we have an equivalent problem
\begin{align}
\mathop {\max}_{\mathbf W, \Theta, \mathbf U} \quad &\hat f \left(\mathbf W, \Theta, \mathbf U\right)=\sum_{k=1}^K{\log \left|\mathbf I +\mathbf U_k\right|}\label{eq75}\\
\operatorname{s.t.}\quad & \Gamma _k=\mathbf U_k,\forall k \in K,\tag{71a}\label{eq76}\\
&\eqref{eq8},\eqref{eq9}\notag.
\end{align} 
%The above optimization can be thought of as an outer optimization over $\left\{\mathbf W,\Theta\right\}$ and an inner optimization over $\mathbf U$ with fixed $\left\{\mathbf W,\Theta\right\}$. 
Note that Problem \ref{eq75} is a convex with respect to $\mathbf U$, so the strong duality holds.
The corresponding Lagrangian function is formulated as 
\begin{align}\label{eq77}
    \hat{\mathcal L}\left(\mathbf W, \Theta, \mathbf U\right)=\sum_{k=1}^K{\log \left|\mathbf I +\mathbf U_k\right|}+\sum_{k=1}^K{\operatorname{Tr}\left(\Upsilon_k\left(\Gamma_k-\mathbf U_k\right)\right)},
\end{align}
where ${\Upsilon_k}\in \mathbb{C}^{M_u \times M_u},\forall k$ are the Lagrangian multipliers associated to the constraints in \eqref{eq76}.
By setting the partial derivative of $\hat{\mathcal L}\left(\mathbf W, \Theta, \mathbf U\right)$ with respect to $\mathbf U _{k},\forall k$  to zeros, we have ${\left( {\mathbf {I} + {\mathbf U_{k}}} \right)^{ - 1}} = \Upsilon_k,\forall k \in K$,
which yields ${\mathbf U_{k}} = {\mathbf I}/{\Upsilon _{k}} - \mathbf {I}$, and substitute $\mathbf U_{k}$ into \eqref{eq77}, we have $\Upsilon _{k}= {\mathbf V_{k}}/{\mathbf {\bar V}_{k}},\forall k \in K$.
Substituting the so-obtained optimal Lagrange multipliers into the expression for the Lagrangian in \eqref{eq77}, and after some modest modification, we obtain $f_1\left(\mathbf W,\Theta, \mathbf U\right)$.
The proof of Proposition \ref{p1} is completed.

\section{Proof of Lemma \ref{lemmaconver}}\label{lemma2proof}
\textcolor{blue}{Let $\rho _i^{\operatorname{q+1}}$ denote the value of the objective function of $\rho ^{\operatorname{q+1}}$ after updating the $i$-th element and fixing the others $\mathcal N -1$ elements of $\theta$ in $q$-th sub-iteration, we have
\begin{align}\label{eqlem}
    \rho^{\operatorname{u}}\le\rho _1^{\operatorname{u+1}}\le\rho _2^{\operatorname{u+1}}\cdots\le\rho _{\mathcal{N}}^{\operatorname{u+1}}=\rho ^{\operatorname{u+1}},
\end{align}
which shows that the value of $f_7\left(\theta\right)$ achieved by Algorithm \ref{a2} increases monotonically. Meanwhile, we have $\theta^{\operatorname{H}}{\mathcal { Z}}\theta \le \alpha^2\mathcal{N}\lambda_{\mathcal { Z}}^{\max}$ and $\operatorname{Re}\left\{\theta^{\operatorname{H}}\omega\right\} \le \alpha\sum_{n=1}^{\mathcal N}\left|\omega_n\right|$,
where $\lambda_{\mathcal { Z}}^{\max}$ is the maximum eigenvalue of $\mathcal { Z}$. The above inequalities yield that the optimal objective value of $f_7\left(\theta\right)$ is upper-bounded by a finite value. Therefore, Algorithm \ref{a2} is guaranteed  to converge.
This completes the proof of Lemma \ref{lemmaconver}.}

% \section*{Acknowledgment}

\ifCLASSOPTIONcaptionsoff
  \newpage
\fi

\bibliography{ref}

\begin{thebibliography}{10}

\bibitem{7917284}
E.~Nayebi, A.~Ashikhmin, T.~L. Marzetta, H.~Yang, and B.~D. Rao, ``Precoding
  and power optimization in cell-free massive mimo systems,'' {\em IEEE
  Transactions on Wireless Communications}, vol.~16, no.~7, pp.~4445--4459,
  2017.

\bibitem{7827017}
H.~Q. Ngo, A.~Ashikhmin, H.~Yang, E.~G. Larsson, and T.~L. Marzetta,
  ``Cell-free massive mimo versus small cells,'' {\em IEEE Transactions on
  Wireless Communications}, vol.~16, no.~3, pp.~1834--1850, 2017.

\bibitem{9354156}
L.~Du, L.~Li, H.~Q. Ngo, T.~C. Mai, and M.~Matthaiou, ``Cell-free massive mimo:
  Joint maximum-ratio and zero-forcing precoder with power control,'' {\em IEEE
  Transactions on Communications}, vol.~69, no.~6, pp.~3741--3756, 2021.

\bibitem{5594708}
D.~{Gesbert}, S.~{Hanly}, H.~{Huang}, S.~{Shamai Shitz}, O.~{Simeone}, and
  W.~{Yu}, ``Multi-cell mimo cooperative networks: A new look at
  interference,'' {\em IEEE Journal on Selected Areas in Communications},
  vol.~28, no.~9, pp.~1380--1408, 2010.

\bibitem{9140329}
M.~{Di Renzo}, A.~{Zappone}, M.~{Debbah}, M.~S. {Alouini}, C.~{Yuen}, J.~{de
  Rosny}, and S.~{Tretyakov}, ``Smart radio environments empowered by
  reconfigurable intelligent surfaces: How it works, state of research, and the
  road ahead,'' {\em IEEE Journal on Selected Areas in Communications},
  vol.~38, no.~11, pp.~2450--2525, 2020.

\bibitem{di2019smart}
M.~Di~Renzo, M.~Debbah, D.-T. Phan-Huy, A.~Zappone, M.-S. Alouini, C.~Yuen,
  V.~Sciancalepore, G.~C. Alexandropoulos, J.~Hoydis, H.~Gacanin, {\em et~al.},
  ``Smart radio environments empowered by reconfigurable ai meta-surfaces: An
  idea whose time has come,'' {\em EURASIP Journal on Wireless Communications
  and Networking}, vol.~2019, no.~1, pp.~1--20, 2019.

\bibitem{9110915}
X.~{Hu}, C.~{Zhong}, Y.~{Zhu}, X.~{Chen}, and Z.~{Zhang}, ``Programmable
  metasurface-based multicast systems: Design and analysis,'' {\em IEEE Journal
  on Selected Areas in Communications}, vol.~38, no.~8, pp.~1763--1776, 2020.

\bibitem{9279253}
M.~Hua, Q.~Wu, D.~W.~K. Ng, J.~Zhao, and L.~Yang, ``Intelligent reflecting
  surface-aided joint processing coordinated multipoint transmission,'' {\em
  IEEE Transactions on Communications}, vol.~69, no.~3, pp.~1650--1665, 2021.

\bibitem{8930608}
Q.~Wu and R.~Zhang, ``Beamforming optimization for wireless network aided by
  intelligent reflecting surface with discrete phase shifts,'' {\em IEEE
  Transactions on Communications}, vol.~68, no.~3, pp.~1838--1851, 2020.

\bibitem{9352948}
Y.~Zhang, B.~Di, H.~Zhang, J.~Lin, C.~Xu, D.~Zhang, Y.~Li, and L.~Song,
  ``Beyond cell-free mimo: Energy efficient reconfigurable intelligent surface
  aided cell-free mimo communications,'' {\em IEEE Transactions on Cognitive
  Communications and Networking}, vol.~7, no.~2, pp.~412--426, 2021.

\bibitem{9039554}
Y.~Yang, B.~Zheng, S.~Zhang, and R.~Zhang, ``Intelligent reflecting surface
  meets ofdm: Protocol design and rate maximization,'' {\em IEEE Transactions
  on Communications}, vol.~68, no.~7, pp.~4522--4535, 2020.

\bibitem{9130088}
Z.~Wang, L.~Liu, and S.~Cui, ``Channel estimation for intelligent reflecting
  surface assisted multiuser communications: Framework, algorithms, and
  analysis,'' {\em IEEE Transactions on Wireless Communications}, vol.~19,
  no.~10, pp.~6607--6620, 2020.

\bibitem{8937491}
B.~Zheng and R.~Zhang, ``Intelligent reflecting surface-enhanced ofdm: Channel
  estimation and reflection optimization,'' {\em IEEE Wireless Communications
  Letters}, vol.~9, no.~4, pp.~518--522, 2020.

\bibitem{9133435}
Q.~Wu and R.~Zhang, ``Joint active and passive beamforming optimization for
  intelligent reflecting surface assisted swipt under qos constraints,'' {\em
  IEEE Journal on Selected Areas in Communications}, vol.~38, no.~8,
  pp.~1735--1748, 2020.

\bibitem{9423652}
S.~Zargari, A.~Khalili, Q.~Wu, M.~Robat~Mili, and D.~W.~K. Ng, ``Max-min fair
  energy-efficient beamforming design for intelligent reflecting surface-aided
  swipt systems with non-linear energy harvesting model,'' {\em IEEE
  Transactions on Vehicular Technology}, vol.~70, no.~6, pp.~5848--5864, 2021.

\bibitem{8741198}
C.~{Huang}, A.~{Zappone}, G.~C. {Alexandropoulos}, M.~{Debbah}, and C.~{Yuen},
  ``Reconfigurable intelligent surfaces for energy efficiency in wireless
  communication,'' {\em IEEE Transactions on Wireless Communications}, vol.~18,
  no.~8, pp.~4157--4170, 2019.

\bibitem{9246254}
H.~Xie, J.~Xu, and Y.-F. Liu, ``Max-min fairness in irs-aided multi-cell miso
  systems with joint transmit and reflective beamforming,'' {\em IEEE
  Transactions on Wireless Communications}, vol.~20, no.~2, pp.~1379--1393,
  2021.

\bibitem{9120476}
G.~Yang, X.~Xu, and Y.-C. Liang, ``Intelligent reflecting surface assisted
  non-orthogonal multiple access,'' in {\em 2020 IEEE Wireless Communications
  and Networking Conference (WCNC)}, pp.~1--6, 2020.

\bibitem{9316920}
G.~Yang, X.~Xu, Y.-C. Liang, and M.~D. Renzo, ``Reconfigurable intelligent
  surface-assisted non-orthogonal multiple access,'' {\em IEEE Transactions on
  Wireless Communications}, vol.~20, no.~5, pp.~3137--3151, 2021.

\bibitem{8955968}
J.~{Gao}, C.~{Zhong}, X.~{Chen}, H.~{Lin}, and Z.~{Zhang}, ``Unsupervised
  learning for passive beamforming,'' {\em IEEE Communications Letters},
  vol.~24, no.~5, pp.~1052--1056, 2020.

\bibitem{8982186}
H.~{Guo}, Y.~{Liang}, J.~{Chen}, and E.~G. {Larsson}, ``Weighted sum-rate
  maximization for reconfigurable intelligent surface aided wireless
  networks,'' {\em IEEE Transactions on Wireless Communications}, vol.~19,
  no.~5, pp.~3064--3076, 2020.

\bibitem{shen2018fractional}
K.~Shen and W.~Yu, ``Fractional programming for communication systems—part i:
  Power control and beamforming,'' {\em IEEE Transactions on Signal
  Processing}, vol.~66, no.~10, pp.~2616--2630, 2018.

\bibitem{shen2018fractional2}
K.~Shen and W.~Yu, ``Fractional programming for communication systems—part
  ii: Uplink scheduling via matching,'' {\em IEEE Transactions on Signal
  Processing}, vol.~66, no.~10, pp.~2631--2644, 2018.

\bibitem{9394419}
M.~A. Saeidi, M.~J. Emadi, H.~Masoumi, M.~R. Mili, D.~W.~K. Ng, and
  I.~Krikidis, ``Weighted sum-rate maximization for multi-irs-assisted
  full-duplex systems with hardware impairments,'' {\em IEEE Transactions on
  Cognitive Communications and Networking}, vol.~7, no.~2, pp.~466--481, 2021.

\bibitem{9388932}
D.~L. Dampahalage, K.~B.~S. Manosha, N.~Rajatheva, and M.~Latva-Aho,
  ``Weighted-sum-rate maximization for an reconfigurable intelligent surface
  aided vehicular network,'' {\em IEEE Open Journal of the Communications
  Society}, vol.~2, pp.~687--703, 2021.

\bibitem{9076830}
G.~Zhou, C.~Pan, H.~Ren, K.~Wang, and A.~Nallanathan, ``Intelligent reflecting
  surface aided multigroup multicast miso communication systems,'' {\em IEEE
  Transactions on Signal Processing}, vol.~68, pp.~3236--3251, 2020.

\bibitem{7547360}
Y.~{Sun}, P.~{Babu}, and D.~P. {Palomar}, ``Majorization-minimization
  algorithms in signal processing, communications, and machine learning,'' {\em
  IEEE Transactions on Signal Processing}, vol.~65, no.~3, pp.~794--816, 2017.

\bibitem{9154244}
Z.~Zhang and L.~Dai, ``Capacity improvement in wideband reconfigurable
  intelligent surface-aided cell-free network,'' in {\em 2020 IEEE 21st
  International Workshop on Signal Processing Advances in Wireless
  Communications (SPAWC)}, pp.~1--5, 2020.

\bibitem{9459505}
Z.~Zhang and L.~Dai, ``A joint precoding framework for wideband reconfigurable
  intelligent surface-aided cell-free network,'' {\em IEEE Transactions on
  Signal Processing}, pp.~1--1, 2021.

\bibitem{9090356}
C.~Pan, H.~Ren, K.~Wang, W.~Xu, M.~Elkashlan, A.~Nallanathan, and L.~Hanzo,
  ``Multicell mimo communications relying on intelligent reflecting surfaces,''
  {\em IEEE Transactions on Wireless Communications}, vol.~19, no.~8,
  pp.~5218--5233, 2020.

\bibitem{5756489}
Q.~{Shi}, M.~{Razaviyayn}, Z.~{Luo}, and C.~{He}, ``An iteratively weighted
  mmse approach to distributed sum-utility maximization for a mimo interfering
  broadcast channel,'' {\em IEEE Transactions on Signal Processing}, vol.~59,
  no.~9, pp.~4331--4340, 2011.

\bibitem{6698378}
M.~Soltanalian and P.~Stoica, ``Designing unimodular codes via quadratic
  optimization,'' {\em IEEE Transactions on Signal Processing}, vol.~62, no.~5,
  pp.~1221--1234, 2014.

\bibitem{9117093}
J.~{Zhang}, Y.~{Zhang}, C.~{Zhong}, and Z.~{Zhang}, ``Robust design for
  intelligent reflecting surfaces assisted miso systems,'' {\em IEEE
  Communications Letters}, vol.~24, no.~10, pp.~2353--2357, 2020.

\bibitem{9180053}
G.~Zhou, C.~Pan, H.~Ren, K.~Wang, and A.~Nallanathan, ``A framework of robust
  transmission design for irs-aided miso communications with imperfect cascaded
  channels,'' {\em IEEE Transactions on Signal Processing}, vol.~68,
  pp.~5092--5106, 2020.

\bibitem{5982443}
P.~Ubaidulla and A.~Chockalingam, ``Relay precoder optimization in mimo-relay
  networks with imperfect csi,'' {\em IEEE Transactions on Signal Processing},
  vol.~59, no.~11, pp.~5473--5484, 2011.

\bibitem{cvx}
M.~Grant and S.~Boyd, ``{CVX}: Matlab software for disciplined convex
  programming, version 2.1.'' \url{http://cvxr.com/cvx}, Mar. 2014.

\bibitem{boyd2004convex}
S.~Boyd, S.~P. Boyd, and L.~Vandenberghe, {\em Convex optimization}.
\newblock Cambridge university press, 2004.

\bibitem{2017Matrix}
X.~Zhang, {\em Matrix Analysis and Applications}.
\newblock Cambridge University Press, 2017.

\bibitem{1634819}
N.~Sidiropoulos, T.~Davidson, and Z.-Q. Luo, ``Transmit beamforming for
  physical-layer multicasting,'' {\em IEEE Transactions on Signal Processing},
  vol.~54, no.~6, pp.~2239--2251, 2006.

\bibitem{7946256}
G.~Cui, X.~Yu, G.~Foglia, Y.~Huang, and J.~Li, ``Quadratic optimization with
  similarity constraint for unimodular sequence synthesis,'' {\em IEEE
  Transactions on Signal Processing}, vol.~65, no.~18, pp.~4756--4769, 2017.

\bibitem{6169188}
S.~Loyka and C.~D. Charalambous, ``On the compound capacity of a class of mimo
  channels subject to normed uncertainty,'' {\em IEEE Transactions on
  Information Theory}, vol.~58, no.~4, pp.~2048--2063, 2012.

\bibitem{9112252}
B.~Lyu, D.~T. Hoang, S.~Gong, D.~Niyato, and D.~I. Kim, ``Irs-based wireless
  jamming attacks: When jammers can attack without power,'' {\em IEEE Wireless
  Communications Letters}, vol.~9, no.~10, pp.~1663--1667, 2020.

\end{thebibliography}
\end{document}